\documentclass[pdflatex,sn-mathphys-num]{sn-jnl}

\usepackage{graphicx}%
\usepackage{multirow}%
\usepackage{amsmath,amssymb,amsfonts}%
\usepackage{amsthm}%
\usepackage{mathrsfs}%
\usepackage[title]{appendix}%
\usepackage{xcolor}%
\usepackage{textcomp}%
\usepackage{manyfoot}%
\usepackage{booktabs}%

\usepackage{algorithm}%
\usepackage{algorithmicx}%
\usepackage{algpseudocode}%
\algrenewcommand\algorithmicrequire{\textbf{Input:}}
\algrenewcommand\algorithmicensure{\textbf{Output:}}

\usepackage{listings}%
\usepackage{times}
\usepackage{paralist}
\usepackage{subcaption}

\graphicspath{ {./Figures/} }
\usepackage{caption}

\theoremstyle{thmstyleone}%
\newtheorem{theorem}{Theorem}
\newtheorem{lemma}{Lemma}
\theoremstyle{thmstyletwo}%
\newtheorem{remark}{Remark}%

\theoremstyle{thmstylethree}%
\newtheorem{definition}{Definition}%

\DeclareMathOperator{\supp}{supp}
\DeclareMathOperator{\diag}{diag}

\newcommand{\nix}[1]{}

\begin{document}

\title[Color codes with domino twists: Construction, logical measurements, and computation]{Color codes with domino twists: Construction, logical measurements, and computation}


\author*[1]{\fnm{Manoj G.} \sur{Gowda}}\email{mgowda.iitk@gmail.com}



\affil*[1]{\orgdiv{Department of Electrical Engineering}, \orgname{Indian Institute of Technology Madras}, \orgaddress{\street{}\city{Chennai} \postcode{600036}, \state{} \country{India}}}




\abstract{ 
Twists are defects that are used to encode and process quantum information in topological codes like surface and color codes. 
Color codes can host three basic types of twists viz., charge-permuting, color-permuting and domino twists.
In this paper, we study domino twists from the viewpoint of computation.
Specifically, we give a systematic construction for domino twists in qubit color codes.
We also present protocols for measurement of logical qubits.
Finally, we show that all Clifford gates can be implemented by braiding twists.}

\keywords{Color codes, twists, braiding, logical measurements}



\maketitle

\section{Introduction}

Color codes are a class of topological codes that were introduced with the objective of performing transversal gates~\cite{Bombin2006}.
Among the various ways of performing computations in color codes are using defects like holes~\cite{Fowler2011,Landahl2011}, twists~~\cite{Kesselring2018,GowdaSarvepalli2021} and lattice surgery~\cite{Landahl2014}.
Twists are a form of defects in the lattice that permute anyons.
Informally, anyons are syndromes (excitations) resulting from the violation of a check operator.
Two labels are needed to completely specify an anyon in color codes viz., charge and color labels.
Color codes support a rich set of anyons in comparison to surface codes.
As a consequence, more anyon permutations and hence more types of twists can be expected in color codes.
Three basic types of twists can be introduced in color codes viz., charge-permuting, color-permuting and domino twists~\cite{Kesselring2018}.
Charge- and color-permuting twists permute the charge and color labels of anyons respectively.
Unlike the charge- or color-permuting twists, domino twists permute a charge label with a color label.
This makes them interesting from a theoretical viewpoint.
The recent advances in the realization of topological codes on hardware~\cite{Mukai2020,Andersen2020,Ryan-Anderson2021,Chen2021,Erhard2021,Marques2022,Hilder2022} provide some impetus to our study of domino twists in color codes.\\

Use of twists for encoding and processing quantum information was first proposed in Ref.~\cite{Kitaev2003}.
Prior studies on twists in surface codes can be found in Refs.~\cite{Bombin2010,Hastings2015,Brown2017,Lavasani2018,GowdaSarvepalli2020}.
Twists in color codes were first studied in Ref.~\cite{Kesselring2018} where all possible twist types in color codes were documented.
Also, the basic twist types with which all other twists can be realized were presented.
Building upon this work, the authors in Ref.~\cite{GowdaSarvepalli2021} presented systematic construction of charge-permuting and color-permuting twists.
Implementing encoded Clifford gates using charge-permuting and color-permuting twists was also explored.
Quantum computation with twists in qudit color codes that permute both charge and color labels of anyons was studied in Ref.~\cite{GowdaSarvepalli2022}.
Some other related work on twists can be found in Refs.~\cite{Bombin2011, Yu-Wen2012}.
Protocols for logical measurements in Calderbank-Shor-Steane (CSS) quantum Low-Density Parity Check (LDPC) codes were proposed in Ref.~\cite{Cohen2022}.
The approach of employing number states and fermion operators for analyzing gates was carried out in Refs.~\cite{Bravyi2002,Bravyi2006} in the context of local Fermionic modes and fractional quantum Hall states respectively.
Recently, this approach was used in the context of superconducting qubits in Ref.~\cite{Scheppe2022}.\\

\noindent \emph{Contribution.}
While Ref.~\cite{Kesselring2018} was the first to introduce domino twists, their computational aspects were not studied.
Computational aspects of charge-permuting and color-permuting twists were studied in Ref.~\cite{GowdaSarvepalli2021}.
A similar study in the case of domino twists seems to be lacking in literature.
This motivates our study of domino twists in color codes from a computational viewpoint.

In this paper, we focus on domino twists in color codes.
Charge-permuting and color-permuting twists are introduced in the lattice by stabilizer modification and lattice modification respectively~\cite{GowdaSarvepalli2021}.
Recall that domino twists permute a color label with a charge label.
Hence, their introduction in lattices is more challenging.
We introduce domino twists in the lattice by suitably modifying the lattice and a careful choice of stabilizer generators.

Many algorithms require qubits to be initialized in specific states like $ \vert 0\rangle $ or $ \vert +\rangle $.
This is done by a logical $Z$ or $X$ measurement respectively.
If the measurement outcome is $-1$, suitable correction Pauli operator is applied.
When the qubits are encoded using a quantum error correcting code, one has to measure the corresponding logical $Z$ or logical $X$ operators.
Logical operators have weight $O(d)$ where $d$ is the minimum code distance.
Good quantum error correcting codes have a large minimum distance.
Therefore, performing a logical measurement by entangling qubits in its support with an ancilla qubit is quite involved and will not be fault-tolerant.
A protocol to measure logical qubits fault-tolerantly in CSS quantum codes was proposed recently in Ref.~\cite{Cohen2022}.
However, color codes with domino twists are non-CSS codes and the protocols presented in Ref.~\cite{Cohen2022} cannot be directly applied.
We present protocols for logical measurements in color codes with domino twists.

Twists carry unpaired Majorana zero modes (MZM)~\cite{Zheng2015} and hence are realizations of non-Abelian anyons~\cite{Bombin2010,Yu-Wen2012}.
Two MZMs can be combined to get a Fermion mode~\cite{Sarma2015}.
Consequently, a pair of twists can be treated as a Fermionic mode.
Previous works~\cite{Brown2017,GowdaSarvepalli2020,GowdaSarvepalli2021} analyze braiding protocols by tracking the evolution of logical operators during braiding.
We analyze the braiding protocols using Fermionic operators and number states.

Our contributions are as follows:
\begin{compactenum}[i)]
\item We expand on the work in Ref.~\cite{Kesselring2018} and present a systematic construction of domino twists in an arbitrary trivalent and three-face-colorable lattice (also called a $2$-colex). The construction is summarized in Theorem~\ref{thm:construction}.
\item We adapt the protocol for logical measurements presented in Ref.~\cite{Cohen2022} to the case of color codes with domino twists. 
This result is summarized in Theorem~\ref{thm:logical-x-measurement}. 
\item We present protocols to implement logical Clifford gates using domino twists.
In doing so, we map a pair of twists created together to a Fermionic mode and analyze twist braiding using Fermion operators and number states.
Theorem~\ref{thm:encoded-Cliffords} summarizes the results on logical Clifford gates.
\end{compactenum}

\vspace{2mm}

\noindent \emph{Organization.} In Section~\ref{sec:background}, we briefly discuss color codes and Pauli operator to string mapping. 
Systematic construction of color codes with domino twists is presented in Section~\ref{sec:domino}.
Logical measurements are discussed in Section~\ref{sec:init-and-meas}.
In Section~\ref{sec:clifford-gates}, we present mapping of color codes with domino twists to Ising anyon model and analyze the braiding rules for implementing Clifford gates using Fermion operator approach.

\section{Background}
\label{sec:background}
In this section we discuss color codes, anyons in color codes and their permutation.
We assume that the reader is familiar with the stabilizer formalism~\cite{Nielsen2010,LidarBrun2013}.
A color code is defined on a $2$-colex which is a trivalent and three face colorable lattice.
Qubits are placed on the vertices and two stabilizer generators are defined on every face $f$:
\begin{equation}
    B_f^X = \prod_{v \in V(f)} X_v, \text{ and  } B_f^Z = \prod_{v \in V(f)} Z_v,
    \label{eqn:color-code-stabilizers}
\end{equation}
where $V(f)$ denotes the set of vertices of a face $f$.

\begin{figure}[htb]
    \centering
    \begin{subfigure}{.45\textwidth}
        \centering
        \includegraphics[scale = .7]{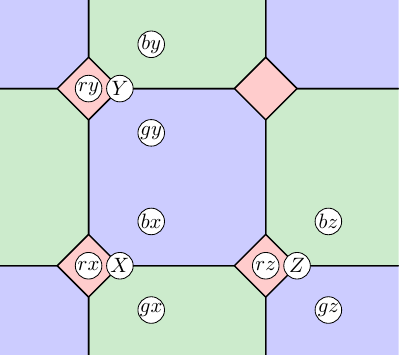}
        \subcaption{}
        \label{fig:anyons}     
    \end{subfigure}
    ~
    \begin{subfigure}{.45\textwidth}
        \centering
        \includegraphics[scale = .7]{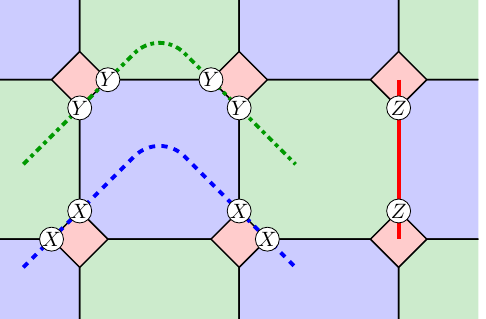}
        \subcaption{}
        \label{fig:strings}
    \end{subfigure}
    \caption{Anyons and strings. (a) Anyons created as result of $X$, $Y$, ad $Z$ errors are shown. Note that a qubit participates in three stabilizers. As a result, an error creates three anyons (syndromes). (b) Strings representing $X$, $Y$, and $Z$ operators are dashed, dash-dotted, and continuous respectively. Note that the end points of a strings are in faces that host an anyon.}
    \label{fig:anyons-strings}
\end{figure}

Suppose that a qubit undergoes $Z$ error.
This error violates an $X$-type stabilizer on faces incident on the qubit (vertex).
As a result, three syndromes, $rx$, $gx$ and $bx$ are created on red, green and blue faces respectively, see Fig.~\ref{fig:anyons}.
The complete set of syndromes in color code is given in Table~\ref{tab:anyons}~\cite{Kesselring2018} and illustrated in Fig.~\ref{fig:anyons}.
\begin{table}[htb]
    \centering
     \caption{Syndromes in color code}
    \begin{tabular}{ccccccccccc}
        \hline 
        \hline
         $rx$ & && & &$ry$ &&&&& $rz$ \\
         $gx$ &&& & & $gy$ &&&&& $gz$\\
         $bx$ &&& & & $by$ &&&&& $bz$ \\
         \hline 
         \hline
    \end{tabular}
    \label{tab:anyons}
\end{table}
Two labels are used to specify a syndrome viz., color labels $r$, $g$ and $b$ indicating red, green and blue faces respectively and Pauli labels $x$, $y$ and $z$ to indicate the type of stabilizer violated. 
Note that a Pauli $Y$ error violates both $X$ and $Z$ type stabilizers defined on a face.
Syndromes are also called anyons sometimes.
We use these terms interchangeably.

\begin{figure}[htb]
    \centering
    \includegraphics[scale = .75]{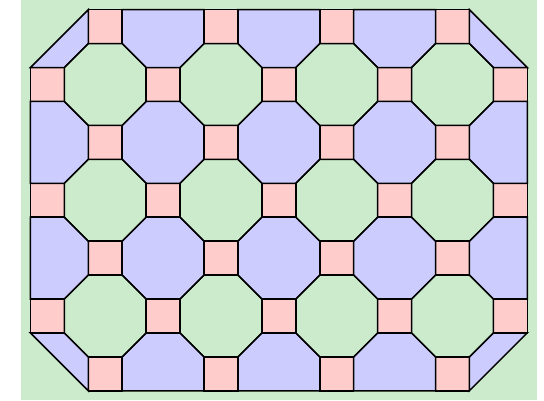}
    \caption{Square-octagon lattice with the unbounded green face. Note that all vertices are trivalent. A stabilizer-counting argument shows that the code defined on the lattice does not encode any logical qubits.}
    \label{fig:basic-lattice}
\end{figure}

An anyon present on a face can be moved to another face by applying suitable Pauli operator.
In doing so, the anyon traces out a path in the lattice.
Such paths are called strings.
In color codes, strings are a collection of edges of the same color in the lattice, see Fig.~\ref{fig:strings}.
A Pauli operator can be associated with a string.
The Pauli operator applied to each qubit to move the anyon around is the Pauli operator associated with the strings.
A string can start and end in the same face.
In such cases, the anyons are created, moved around and annihilated.
Pauli operators corresponding to such strings are either stabilizers or logical operators.
On the hand, strings that start and end in different faces only move the anyon apart.
The operators corresponding to such strings form the set of detectable errors.
For a detailed discussion on strings, see Ref.~\cite{GowdaSarvepalli2021}.
In this paper, we consider $2$-colexes with boundary. 
We assume that the unbounded face on the boundary is of green color as shown in Fig.~\ref{fig:basic-lattice}.

\section{Construction}
\label{sec:domino}

In this section, we discuss the creation and movement of domino twists in arbitrary $2$-colexes. 
We generalize the construction of domino twists in Ref.~\cite{Kesselring2018} to an arbitrary $2$-colex. 
Before moving on to construction, we need a notion of how far apart two domino twists are.
This is formalized in the definition below.\\

\begin{definition}[Domino twist separation]
Separation between a domino twist pair is the length of the smallest path between the vertices corresponding to twists in the dual lattice.
\end{definition}

\begin{figure}[htb]
    \centering
    \begin{subfigure}{0.45\textwidth}
        \centering
        \includegraphics[scale = .95]{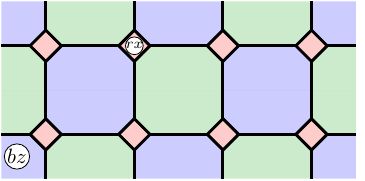}
        \subcaption{}
        \label{fig:domino-intuition-1}
    \end{subfigure}
    ~
    \begin{subfigure}{0.45\textwidth}
        \centering
        \includegraphics[scale = .95]{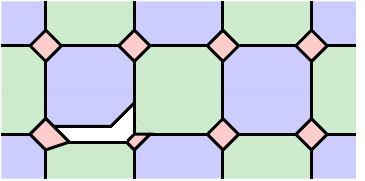}
        \subcaption{}
        \label{fig:domino-intuition-2}
    \end{subfigure}
    ~
    \begin{subfigure}{0.45\textwidth}
        \centering
        \includegraphics[scale = .95]{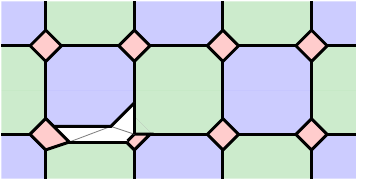}
        \subcaption{}
        \label{fig:domino-intuition-5}
    \end{subfigure}
    ~
    \begin{subfigure}{0.45\textwidth}
        \centering
        \includegraphics[scale = .95]{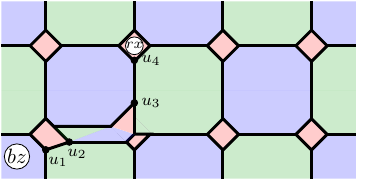}
        \subcaption{}
        \label{fig:domino-intuition-6}
    \end{subfigure}
    \caption{Intuition. (a) The desired transformation to be achieved is $bz \longrightarrow rx$. Note that this transformation cannot be achieved without modifying the lattice. (b) Introduce a face as shown. This transformation leaves a red pentagon face and a green nonagon face. (c) Partition the face virtually as shown. This creates three half-bricks. (d) Color the half bricks as shown. With a suitable choice of stabilizers on twists and a combination of two half-bricks, the desired transformation can be achieved.}
    \label{fig:domino-intuition}
\end{figure}

\noindent \emph{\textbf{Intuition} --} Twists are introduced in the lattice by modifying stabilizers or the structure of the lattice~\cite{GowdaSarvepalli2021} or both lattice and stabilizers as in the case of charge-and-color-permuting twists~\cite{GowdaSarvepalli2022}.
Recall that domino twists permute the charge label of an anyon with its color label~\cite{Kesselring2018}, for instance, $r \longleftrightarrow z$, $g \longleftrightarrow y$, $b \longleftrightarrow x$.
This specific permutation leaves the anyons $rz$, $gy$ and $bx$ unaltered.
Domino twists cannot be introduced in the lattice by modifying stabilizers or the structure of the lattice or both.

Suppose we want the permutation $r \longleftrightarrow z$, $g \longleftrightarrow y$, and $b \longleftrightarrow x$. 
For instance, this permutes $bz$ to $rx$, see Fig.~\ref{fig:domino-intuition-1}.
Observe that in a $2$-colex, only faces of the same color are connected by an edge.
Therefore, syndromes can be moved only between faces of the same color.
This makes it difficult to achieve the transformation $bz \longrightarrow rx$. 
Consider the modification of the lattice as given in Fig.~\ref{fig:domino-intuition-2}.
An additional face is introduced and as a result red square face is now a pentagon and the green face is a nonagon.
Partition the face virtually as shown  in Fig.~\ref{fig:domino-intuition-5} and color the partitions as shown in Fig.~\ref{fig:domino-intuition-6}. 
Each partition is called a half-brick.
A stabilizer is defined on the combination of two adjacent half-bricks.
A half-brick can participate in the definition of more than one stabilizer.
Now suppose that we want to move the anyon $bz$ across the domain wall. 
There is now a path with edges $(u_1,u_2)$ and $(u_3,u_4)$ between the blue and red faces.
With suitable choice of stabilizers on twists and combination of a pair of bricks and Pauli operators on vertices $u_i$, $1 \le i \le 4$, one can transform the anyon $bz$ to $rx$.

\begin{figure}[htb]
    \centering
    \begin{subfigure}{.45\textwidth}
       \centering
       \includegraphics[width = 5cm,height = 1.25cm]{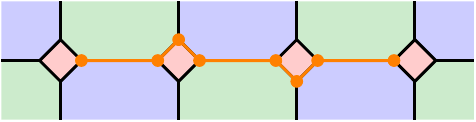}
       \subcaption{}
       \label{fig:dt-creation-1}
    \end{subfigure}
    ~
    \begin{subfigure}{.45\textwidth}
       \centering
       \includegraphics[width = 5cm,height = 1.25cm]{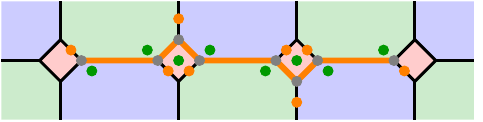}
       \subcaption{}
       \label{fig:dt-creation-2-1}
    \end{subfigure}
    ~
    \begin{subfigure}{.45\textwidth}
       \centering
       \includegraphics[width = 5cm,height = 1.25cm]{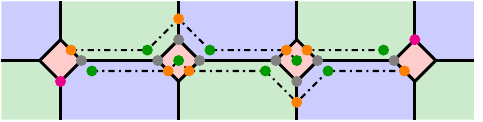}
       \subcaption{}
       \label{fig:dt-creation-2-2}
    \end{subfigure}
    ~
        \begin{subfigure}{.45\textwidth}
       \centering
       \includegraphics[width = 5cm,height = 1.25cm]{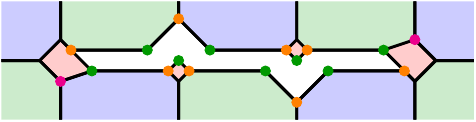}
       \subcaption{}
       \label{fig:dt-creation-3}
    \end{subfigure}
    ~
        \begin{subfigure}{.45\textwidth}
       \centering
       \includegraphics[width = 5cm,height = 1.25cm]{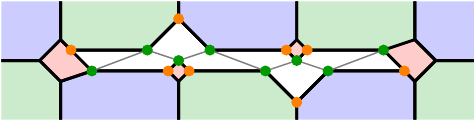}
       \subcaption{}
       \label{fig:dt-creation-4}
    \end{subfigure}
    ~
    \begin{subfigure}{.45\textwidth}
       \centering
       \includegraphics[width = 5cm,height = 1.25cm]{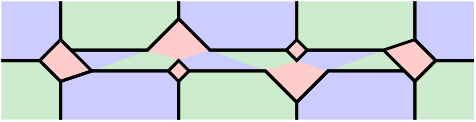}
       \subcaption{}
       \label{fig:dt-creation-5}
    \end{subfigure}
    \caption{Illustration of domino twist creation in the square-octagon lattice. (a) The path $\pi$ chosen for creation of domino twists. This path does not have more than two edges from a face. (b) Vertex doubling along the path $\pi$ chosen in Fig.~\ref{fig:dt-creation-1}. (c) Paths are created by connecting vertices on either side of the path $\pi$ and deleting the vertices along $\pi$. (d) Remove the vertices along the path $\pi$ and connect the ends of the two paths created in the previous step. Also, connect the two-valent vertices, shown in magenta, to the nearest vertex $v_{\ast}$. (e) Create half-bricks by connecting vertices as shown. (f) Color the half-bricks as shown. The red pentagon faces are twists.}
    \label{fig:dt-creation}
\end{figure}

\subsection{Creation and movement}
Systematic construction of domino twists in a $2$-colex is given in Algorithm~\ref{alg:domino-twists-creation}.
The algorithms for twist creation and movement are separated. 
The reason being, once a pair of twists is created but not moved apart, the original lattice changes and the set of edges chosen for twist creation is modified.
The algorithm for twist creation is illustrated in Fig.~\ref{fig:dt-creation}.
We use the following notation to distinguish the vertices created by doubling in Algorithm~\ref{alg:domino-twists-creation}.
Vertices introduced on edges and faces are denoted as  $v_{\bullet}$ and $v_{\star}$ respectively. 
The corresponding sets are denoted by $\mathsf{V}_{\bullet}$ and $\mathsf{V}_{\ast}$ respectively.
We also denote the vertex introduced on an edge (face) due to a vertex $u$ as $v_{\bullet}(u)$ ($v_{\ast}(u)$).
We use the term virtual path in Algorithm~\ref{alg:domino-twists-creation} because the edges along this path are not actual edges of the lattice.
They only serve to create half-bricks.

\begin{algorithm}[H]
    \caption{Algorithm to create $k$ pairs of domino twists in a $2$-colex.}
    \label{alg:domino-twists-creation}
    \begin{flushleft}
    \algorithmicrequire{ $2$-colex, number of twist pairs $k \ge 1$, separation between twists $\ell \ge 1$.}\\
    \algorithmicensure{ Trivalent lattice with $k$ pairs of twists.}
    \end{flushleft}
\begin{algorithmic}[1]
       \State{Select $k$ continuous paths $\pi_i$ in the lattice such that no two paths overlap and that a face contributes no more than two edges to a path, see Fig.~\ref{fig:dt-creation-1}.}
       \For{$i = 1$ to $k$}
       \State{Double the vertices along the path $\pi_i$ as shown in Fig.~\ref{fig:dt-creation-2-1} such that one of the new vertices is on an edge that is not in the path $\pi_i$ and the other in a face. No two doubled vertices should be in the same face.}
       \State{Connect the new vertices so that the path formed by new edges is parallel to the chosen path $\pi_i$ and has vertices alternating between $\mathsf{V}_{\ast}$ and $\mathsf{V}_{\bullet}$, see Fig.~\ref{fig:dt-creation-2-2}. 
        \State{Delete the vertices on the path $\pi_i$. This results in two-valent vertices, see Fig.~\ref{fig:dt-creation-2-2}. Connect the two-valent vertices to the nearest doubled vertex $v_{\ast}$, see Fig.~\ref{fig:dt-creation-3}.}
       \State Connect the vertices $v_{\ast}(u)$ and $v_{\bullet}(u)$ corresponding to a terminal vertex $u$ of the path $\pi_i$, see Fig.~\ref{fig:dt-creation-3}.} 
        \State{Connect the vertices in $\mathsf{V}_{\ast}$ by a virtual edge, see Fig.~\ref{fig:dt-creation-4}.}
        \Comment{This creates half-bricks.} 
       \State{Color of a half-brick is the color of the face (which is not a twist) with which it shares exactly a vertex, see Fig.~\ref{fig:dt-creation-5}. }
       \EndFor
    \end{algorithmic}
\end{algorithm}

\noindent \emph{Twist movement.}
The algorithm for twist movement is given in Algorithm~\ref{alg:domino-twists-movement} and is illustrated in Fig.~\ref{fig:dt-shift}.

\begin{figure}[htb]
    \centering
    \begin{subfigure}{.45\textwidth}
        \centering
        \includegraphics[height = 3cm, width = 6cm]{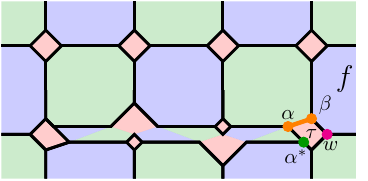}
        \subcaption{}
        \label{fig:dt-shift-1}
    \end{subfigure}
    ~
    \begin{subfigure}{.45\textwidth}
        \centering
        \includegraphics[height = 3cm, width = 6cm]{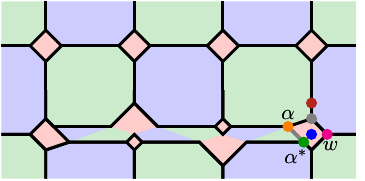}
        \subcaption{}
        \label{fig:dt-shift-2}
    \end{subfigure}
    ~
    \begin{subfigure}{.45\textwidth}
        \centering
        \includegraphics[height = 3cm, width = 6cm]{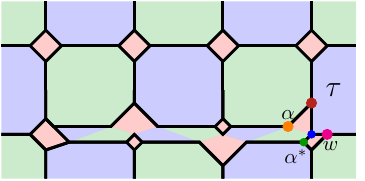}
        \subcaption{}
        \label{fig:dt-shift-3}
    \end{subfigure}
     ~
    \begin{subfigure}{.45\textwidth}
        \centering
        \includegraphics[height = 3cm, width = 6cm]{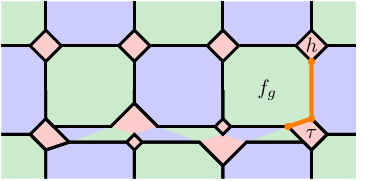}
        \subcaption{}
        \label{fig:dt-shift-4}
    \end{subfigure}
     ~
    \begin{subfigure}{.45\textwidth}
        \centering
        \includegraphics[height = 3cm, width = 6cm]{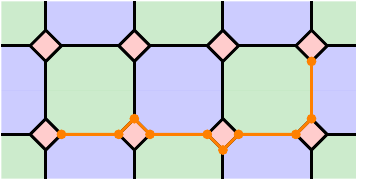}
        \subcaption{}
        \label{fig:dt-shift-5}
    \end{subfigure}
    \caption{Illustration of twist movement. (a) The face to which twist $\tau$ should be moved is marked $f$. The edge incident on the face $f$ which is used for vertex doubling is shown in orange. (b) Vertex of the chosen edge $\beta$ is doubled. The vertices $v_{\bullet}(\beta)$ and $v_{\ast}(\beta)$ are shown as red and blue vertices respectively. The edge $(\alpha, \alpha^{\ast})$is removed. (c) The new vertices introduced are connected as shown. Also, the two-valent vertices resulting from vertex deletion are also connected to the new vertices. Virtual edge is introduced and the resulting half-brick is colored according to Algorithm~\ref{alg:domino-twists-creation}. (d) The path along which the twist $\tau$ cannot be moved. (e) Moving twist to face $h$ is equivalent to applying Algorithm~\ref{alg:domino-twists-creation} to the path shown. This path violates the condition in Algorithm~\ref{alg:domino-twists-creation} that a path should not have more than two common edges with a face.}
    \label{fig:dt-shift}
\end{figure}

\begin{algorithm}[htb]
    \caption{Algorithm to move domino twists.}
    \label{alg:domino-twists-movement}
    \begin{flushleft}
    \algorithmicrequire{ Lattice with at least a pair of domino twists separated by distance $\ell$.}\\
    \algorithmicensure{ Lattice with a pair of domino twists separated by distance $\ell + 1$.}
    \end{flushleft}
    
\begin{algorithmic}[1]
     \State{Choose an edge $e = (\alpha,\beta)$, such that the edge $e$ bounds the twist face, vertex $\alpha$ is incident on the domain wall, and the vertex $\beta$ is common to the twist and the face to which twist is to be moved, see Fig.~\ref{fig:dt-shift-1}.}
     \State{Double the vertex $\beta$ such that $v_{\ast}(\beta)$ lies in the twist face. Remove the edge $(\alpha, \alpha^{\ast})$ and delete the vertex $\beta$, see Fig.~\ref{fig:dt-shift-2}. }
    \State{Introduce new edges $(\alpha,v_{\bullet}(\beta))$,$(v_{\bullet}(\beta),v_{\ast}(\beta))$, $(\alpha^{\ast},v_{\ast}(\beta))$, and $(v_{\ast}(\beta),w)$, see Fig.~\ref{fig:dt-shift-3}.}
    \State{Introduce virtual edge $(\alpha,v_{\ast}(\beta))$ and color the resulting half-brick according to Algorithm~\ref{alg:domino-twists-creation}, see Fig.~\ref{fig:dt-shift-3}.}
    \end{algorithmic}
\end{algorithm}

\begin{remark}
    Note that the twist $\tau$ cannot be moved to the face $h$ shown in Fig.~\ref{fig:dt-shift-4}.
    Moving twist $t$ to face $h$ is equivalent to applying Algorithm~\ref{alg:domino-twists-creation} to the path shown in Fig.~\ref{fig:dt-shift-5}. Note that this path violates the condition in statement 1 of Algorithm~\ref{alg:domino-twists-creation} as the face $f_g$ shares three common edges with the path.
\end{remark}

\subsection{Stabilizers}
Three types of faces can be identified in a color code lattice with domino twists: normal faces, twists, and bricks.

\vspace{2mm}

\noindent \emph{Normal faces}. We call the faces with even number of edges as normal faces.
Two stabilizers are defined on a normal face $f$:
\begin{equation}
    B_{f}^X = \prod_{v \in V(f) }X_v, \text{ and }  B_{f}^Z = \prod_{v \in V(f) }Z_v.
    \label{eqn:normal-stabilizers}
\end{equation}

\begin{figure}[htb]
    \centering
    \begin{subfigure}{.225\textwidth}
        \centering
        \includegraphics[scale = 1.05]{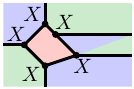}
        \subcaption{}
        \label{fig:domino-twist-stabilizers-Z}
    \end{subfigure}
    ~
    \begin{subfigure}{.225\textwidth}
        \centering
        \includegraphics[scale = 1.05]{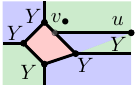}
        \subcaption{}
        \label{fig:domino-twist-stabilizers-Y}
    \end{subfigure}
    ~
    \begin{subfigure}{.225\textwidth}
        \centering
        \includegraphics[scale = .75]{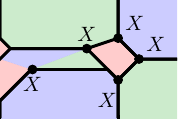}
        \subcaption{}
        \label{fig:domino-stabilizer-x-2}
    \end{subfigure}
    ~
    \begin{subfigure}{.225\textwidth}
        \centering
        \includegraphics[scale = .75]{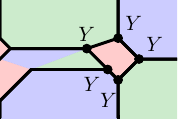}
        \subcaption{}
        \label{fig:domino-stabilizer-y-2}
    \end{subfigure}
    ~
    \begin{subfigure}{.9\textwidth}
        \centering
        \includegraphics[scale = 1]{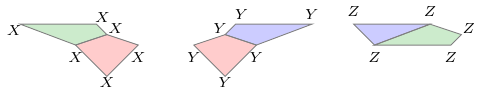}
        \subcaption{}
        \label{fig:domino-bricks-all}
    \end{subfigure}
    \caption{Stabilizer assignment. (a) $X$-type stabilizer defined on domino with full support on twist face. (b) $Y$-type stabilizer defined on domino twist with partial support on twist face.  Note that we cannot define two stabilizers as they will anticommute. Therefore, one of the stabilizers, namely $X$, is defined on the vertices of the twist and the other has partial support on the twist face. To be precise, the size of support for $Y$ stabilizer is an even number, thus resulting in commuting stabilizers. Stabilizer generators on the other twist are shown in (c) and (d). (e) Domino brick stabilizers for the $D_1$ twist.}
    \label{fig:domino-twist-stabilizers}
\end{figure}

\noindent \emph{Domino bricks}.
In order to get the desired transformation of anyons, we need to define some additional stabilizers.
These stabilizers are not defined on faces but rather on domino bricks.
There are three types of domino bricks depending on the color of the component half-bricks: red-green, red-blue, and blue-green.
The type of stabilizers defined depends on the type of half bricks combined~\cite{Kesselring2018}.
An $X$, $Y$ or $Z$ type stabilizer is assigned to the domino bricks depending on the kind of anyon transformation desired.
Consequently, there are six choices for stabilizers for domino bricks and are given in Table~\ref{tab:stab}.

In this paper, we consider $D_1$ twists whereas $D_6$ twists were studied in Ref.~\cite{Kesselring2018}.
The charge-color permutation effected by the $D_1$ twist is as follows: $r\longleftrightarrow z,  g \longleftrightarrow y, b \longleftrightarrow x$.

\begin{table}[htb]
\caption{Domino bricks stabilizers}
\centering
    \begin{tabular}{p{3cm}  p{.5cm} p{.5cm} p{.5cm} p{.5cm} p{.5cm} p{.5cm}}
        \hline
        \hline
         Color pair/Twist type  & $D_1$ & $D_2$ &$D_3$ &$D_4$ & $D_5$ &$D_6$\\
        \hline
        red, green & $X$ & $X$ &$Y$ & $Y$ & $Z$ &$Z$  \\
        \hline
        red, blue & $Y$ & $Z$& $X$ & $Z$& $X$ &$Y$ \\
        \hline
        green, blue & $Z$ & $Y$ & $Z$ & $X$ &  $Y$& $X$\\
        \hline
        \hline
    \end{tabular}
    \label{tab:stab}
\end{table}

If a color label $c$ is permuted with a Pauli label $p$, the domino bricks having half-bricks of color $c$ cannot have stabilizers of $p$-type.
For instance, if the required transformations are $r \longleftrightarrow z$, $g \longleftrightarrow y$ and $b \longleftrightarrow x$ then a domino brick with a red half-brick as a component cannot have $Z$-type stabilizer.
The reason is as follows:
let $(u,v)$ be the edge incident on a red face with an $rx$ anyon and a red half-brick.
The operator $Z_u Z_v$ annihilates the anyon $rx$ on the red face without creating it on the half-brick.
Similarly, one can argue that green and blue half bricks cannot have $Y$- and $X$-type stabilizers.
Hence, the consistent stabilizer assignment is red-green domino brick is assigned $X$-type stabilizer, red-blue and blue-green are assigned $Y$- and $Z$-type stabilizers respectively:
\begin{equation}
    B_{rg}^d = \!\!\!\!\!\! \prod_{v \in \supp\{ d \} } \!\!\!\!\!\! X_v, \quad B_{rb}^d = \!\!\!\!\!\! \prod_{v \in \supp\{ d \} } \!\!\!\!\!\!  Y_v, \quad B_{bg}^d = \!\!\!\!\!\! \prod_{v \in \supp\{ d \} }   \!\!\!\!\!\! Z_v.
    \label{eqn:domino-brick-stabilizers}
\end{equation}
where $d$ indicates domino brick. 
The domino brick stabilizer generators are shown in Fig.~\ref{fig:domino-bricks-all}.
\newline

\noindent \emph{Twists}.
Two stabilizer generators are defined on twist faces.
For one of the stabilizer generators, vertices of the twist face form the support and for the other stabilizer generator an even number of vertices are taken from twist face and the other vertex is the neighbor of the twist face vertex not in the support, see Figs.~\ref{fig:domino-twist-stabilizers-Y} and~\ref{fig:domino-stabilizer-y-2}.
The Pauli operator type for stabilizer generator with support on all vertices of twist face is decided by the domino brick stabilizer that shares a common vertex with twist, see Figs.~\ref{fig:domino-twist-stabilizers-Z} and ~\ref{fig:domino-stabilizer-x-2}.
The Pauli operator type for stabilizer generator with partial support on vertices of twist face is decided by the domino brick stabilizer that has the vertex $u$ shown in Fig.~\ref{fig:domino-twist-stabilizers-Y} as common vertex (in this case red-blue domino brick).
Stabilizers for the twist face shown in Figs.~\ref{fig:domino-twist-stabilizers-Z} to~\ref{fig:domino-stabilizer-y-2} respectively are as below:
\begin{equation}
    B_{\tau,X}^{(1)} = \!\!\! \prod_{v \in V(\tau) } \!\!\! X_v, \text{  } B_{\tau,Y}^{(1)} = \!\!\! \prod_{v \in U(\tau) } \!\!\! Y_v, \text{  } B_{\tau,X}^{(2)} = \!\!\! \prod_{v \in U(\tau) }\!\!\! X_v, \text{  } B_{\tau,Y}^{(2)} = \!\!\! \prod_{v \in V(\tau) } \!\!\! Y_v,
    \label{eqn:twist-stabilizers}
\end{equation}
where $U(\tau) = \{ V(\tau) \setminus v_{\bullet} \} \cup \{u\}$ and $u = N(v_{\bullet}) \setminus V(\tau)$, see Fig.~\ref{fig:domino-twist-stabilizers}.
Note that the $X$ and $Y$ stabilizer definitions are reversed for the other twist face and are shown in Figs.~\ref{fig:domino-stabilizer-x-2} and ~\ref{fig:domino-stabilizer-y-2}.\\

\noindent \emph{Stabilizer commutation.}
The following pairs share exactly an edge: two adjacent normal faces, a normal face and a twist, a twist and a domino brick, and two adjacent domino bricks.
Therefore, the corresponding stabilizer generators commute.
A twist face and a domino brick can have a common vertex.
By construction, the domino brick and the twist have the same Pauli operator on the common vertex and therefore the stabilizer generators defined on them commute.
The stabilizers $B_{\tau,X}^{(i)}$ and $B_{\tau,Y}^{(i)}$ intersect exactly at even number of vertices and therefore commute.
Hence, the operators given in Equations~\eqref{eqn:normal-stabilizers},~\eqref{eqn:domino-brick-stabilizers},~\eqref{eqn:twist-stabilizers} are valid stabilizer generators.
The stabilizer group $\mathcal{S}$ is specified as below:
\begin{equation}
    \mathcal{S} = \left\langle B_f^X,B_f^Z, B_{\tau,P}^{(1)}, B_{\tau,P}^{(2)}, B_{c c^\prime}^d \vert P \in \{ X,Y \}, c,c^\prime \in \{r,g,b\},  c \neq c^\prime \right\rangle.
\end{equation}
The action of domino twists on anyons $gx$ and $bx$ as they cross the domain wall is shown in Fig.~\ref{fig:twist-action}.
\begin{figure}[htb]
    \centering
    \includegraphics{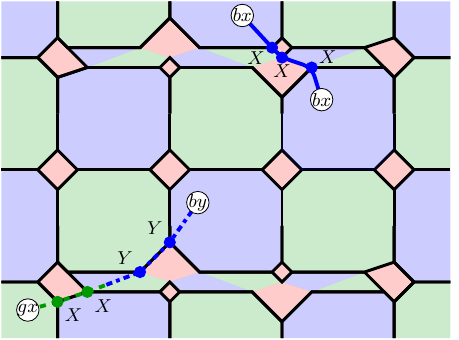}
    \caption{Action of domino twist: the anyon $gx$ is permuted to $by$ upon crossing the domain wall whereas the anyon $bx$ is left invariant.}
    \label{fig:twist-action}
\end{figure}

\nix{Stabilizer generators defined in the preceding paragraphs commute as any two adjacent faces or a domino brick and a face share exactly an edge. 
Also, any two domino bricks share even number of vertices.
Hence commutation follows.
}

\noindent \emph{Stabilizer constraints.} Stabilizers  satisfy the following constraints:
\begin{subequations}
    \begin{eqnarray}
        \prod_{f \in \mathsf{F}_{bg}} B_f^Z \prod_{d \in \mathsf{B}_{bg}} B_{bg}^d &=& I,\\
        \prod_{f \in \mathsf{F}_{rg}} B_f^X \prod_{d \in \mathsf{B}_{rg}} B_{rg}^d \prod_{\tau} B_{\tau,X} &=& I,\\
        \prod_{f \in \mathsf{F}_{rb}} B_f^X B_f^Z \prod_{d \in \mathsf{B}_{rb}} B_{rb}^d \prod_{\tau} B_{\tau,Y} &=& I,     
    \end{eqnarray}
     \label{eq:stab-dep}
\end{subequations}
where $\mathsf{B}_{cc^{\prime}}$ the set of domino bricks formed by combining half-bricks of color $c$ and $c^{\prime}$, $\mathsf{F}_{cc^{\prime}} = \mathsf{F}_c \cup \mathsf{F}_{c^{\prime}}$ and $\tau$ are the twists.
Note that we include the stabilizer defined on the unbounded green face in the above product.
The above equations suggest that three stabilizers are redundant.
We choose two redundant stabilizers to be the ones defined on the unbounded green face and the other one a stabilizer on a red face.
Using this result, we derive the number of encoded qubits in a lattice with domino twists.
Note that, by construction, domino twists are always created in pairs.
As a result, the number of domino twists is always an even number. \\

\begin{lemma}[Encoded qubits]
A 2-colex with $t \ge 2$ domino twists encodes $k = \frac{t}{2} - 1$ logical qubits.
\label{lm:encoded-qubits-domino}
\end{lemma}
\begin{proof}
    See Appendix~\ref{sec:proof-encoded-qubits}. 
\end{proof}

\begin{figure}[htb]
    \centering
    \begin{subfigure}{.45\textwidth}
        \centering
        \includegraphics[width = 5.336cm, height = 4cm]{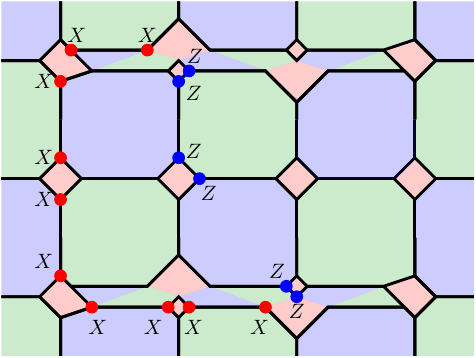}
        \subcaption{}
        \label{fig:domino-LO-X}
    \end{subfigure}
    ~
    \begin{subfigure}{.45\textwidth}
        \centering
        \includegraphics[width = 5.336cm, height = 4cm]{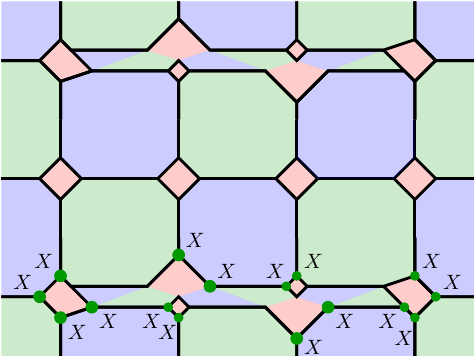}
        \subcaption{}
        \label{fig:domino-LO-Z}
    \end{subfigure}
    ~
    \begin{subfigure}{.45\textwidth}
        \centering
        \includegraphics[scale = .65]{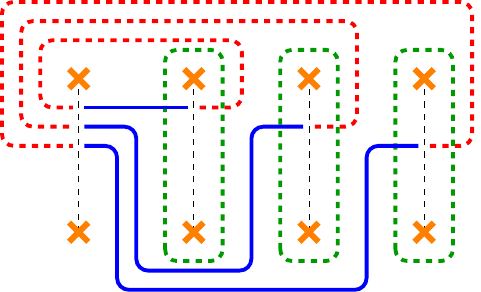}
        \subcaption{}
        \label{fig:logical-four-pairs-twists}
    \end{subfigure}
    ~
    \begin{subfigure}{.45\textwidth}
        \centering
        \includegraphics[scale = .85]{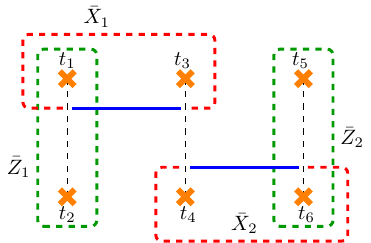}
        \subcaption{}
        \label{fig:lo-domino}
\end{subfigure}
\caption{Logical operators: (a) Logical $X$ operator depicted on lattice. (b) Logical $Z$ operator depicted on lattice.  Note that logical $Z$ and $X$ operators have different Pauli operators on a single vertex in their common support. Hence, they anticommute. (c) Logical operators for $(t/2) - 1$ logical qubits. Green strings are logical $Z$ operators and red-blue strings are logical $X$ operators. (d) Logical operators for $\lfloor t / 3 \rfloor$ encoding.}
\label{fig:domino-LO}
\end{figure}

\subsection{Logical Operators}
Logical operators commute with all stabilizer generators and are not in the stabilizer group.
For ease of representation, abstract the lattice as in Refs.~\cite{GowdaSarvepalli2020,GowdaSarvepalli2021} and consider the twists and domain walls.
Pauli operators are represented as strings.
An $X$ and $Z$-string are represented as dashed and continuous string respectively.
Logical operators on the lattice are depicted in Fig.~\ref{fig:domino-LO-X} and Fig.~\ref{fig:domino-LO-Z}.
The strings corresponding to these operators encircle a even number of twists, see Fig.~\ref{fig:logical-four-pairs-twists}.
The logical operators for $(t / 2)  - 1$ logical qubits are shown in Fig.~\ref{fig:logical-four-pairs-twists}. 
The red-blue strings correspond to logical $X$ operators.
Observe that the weight of $X$ logical operators increases with the number of logical qubits.
A similar observation was made in the context of surface codes with twists in Ref.~\cite{GowdaSarvepalli2020}.
Therefore, we choose the encoding shown in Fig.~\ref{fig:lo-domino}.
This scheme encodes lesser logical qubits while ensuring the same weight for all logical $X$ operators.
The number of logical qubits encoded are $\lfloor t / 3 \rfloor$.
Remaining $(t/2) - 1 - \lfloor t / 3 \rfloor$ logical qubits are treated as gauge qubits.
The construction is summarized in Theorem~\ref{thm:construction}.\\

\begin{theorem}[Code construction]
A 2-colex with $t$ domino twists employing $\lfloor t / 3 \rfloor$ encoding defines a subsystem code with $\lfloor t / 3 \rfloor$ logical qubits and  $ t/2 - 1 - \lfloor t / 3 \rfloor$ gauge qubits.
\label{thm:construction}
\end{theorem}

\vspace{.15in}

\noindent This concludes our construction of domino twists in an arbitrary $2$-colex.
We now proceed to the measurement of logical qubits encoded using domino twists.

\section{Logical measurements}
\label{sec:init-and-meas}

In this section, we present protocols to perform logical $Z$ and $X$ measurements.
Logical $Z$ measurement can be implemented using the protocol presented in Ref.~\cite{Cohen2022}.
However, the same cannot be used for performing logical $X$ measurements.
We adapt the protocols presented in Ref.~\cite{Cohen2022} to perform logical $X$ measurements.
We now briefly review the protocol presented in Ref.~\cite{Cohen2022} for performing single qubit logical measurement.

\begin{figure}[htb]
    \centering
    \begin{subfigure}{.225\textwidth}
        \centering
        \includegraphics[scale = .75]{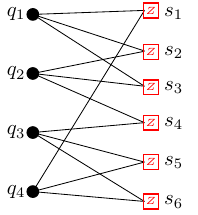}
        \subcaption{}
        \label{fig:tanner_graphs}
    \end{subfigure}
    ~
    \begin{subfigure}{.225\textwidth}
        \centering
        \includegraphics[scale = .75]{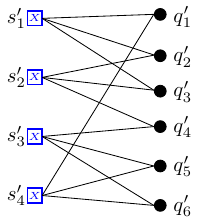}
        \subcaption{}
        \label{fig:dual-tanner_graphs}
    \end{subfigure}
    ~
    \begin{subfigure}{.45\textwidth}
        \centering
        \includegraphics[scale = .85]{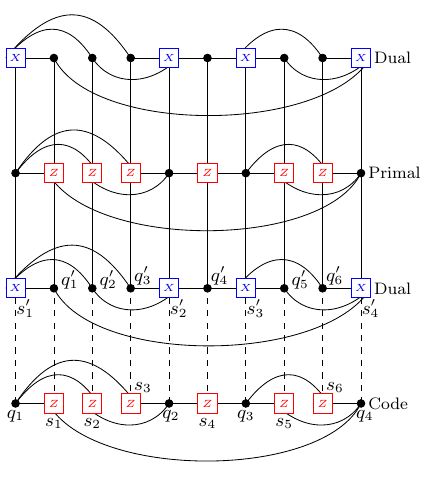}
        \subcaption{}
        \label{fig:log-Z-meas}
    \end{subfigure}
    \caption{Tanner graphs. The sub-Tanner graph corresponding to a logical operator $\bar L$ with support on qubits $q_i$, $1 \le i \le 4$ is shown in (a). The dual of the Tanner graph in (a) is shown in (b). It is obtained by replacing checks with qubits and qubits with $X$-type checks. The Tanner graph for measuring operator $\bar L$ is shown in (c). The bottom layer corresponds to the code and the rest to the ancilla. Notice that a primal layer is sandwiched between two dual layers.}
    \label{fig:tanner-graphs}
\end{figure}

Suppose that the logical operator we like to measure is $\bar{L} = \prod_{i = 1}^4 X_{q_i}$.
The qubits in support of $\bar{L}$ are $q_i$, $1 \le i \le 4$.
Let the $Z$-type stabilizers that share common qubits with $\bar L$ be $s_i$, $1 \le i \le 6$.
A Tanner graph corresponding to the qubits in support of the logical operator $\bar L$ is shown in Fig.~\ref{fig:tanner_graphs}.
Circles represent qubits (bit nodes) and squares represent $Z$-type stabilizers (also called check nodes).
There is an edge between a bit node $q_i$ and check node $s_j$ in the Tanner graph if and only if the bit node $q_i$ participates in the check $s_j$.
Observe that the Tanner graph is bipartite; no two circles or squares are connected by an edge.
Dual of a Tanner graph is constructed by replacing circles with squares and vice versa.
If the squares of the Tanner graph are $Z$-type stabilizers, then the dual Tanner graph will have $X$-type stabilizers for squares, see Fig.~\ref{fig:dual-tanner_graphs}.
The connections between the two sets remain unchanged.

An ancilla system with specific stabilizer structure is constructed to perform logical measurements.
The construction of ancilla lattice for logical measurement, as described in Ref.~\cite{Cohen2022}, is shown in Fig.~\ref{fig:log-Z-meas}.
The first layer of the ancilla is the dual Tanner graph with $X$-type stabilizers.
This is followed by an alternating sequence of primal and dual Tanner graphs, ending in a dual Tanner graph.
One takes $w - 1$ primal Tanner graphs sandwiched between $w$ dual Tanner graphs where $w$ is the weight of the logical operator being measured.
For brevity, Fig.~\ref{fig:log-Z-meas} shows two dual Tanner graphs and a primal Tanner graph.
The qubits and $X$-type stabilizers of the dual Tanner graph in the first layer are connected to the corresponding $Z$-type stabilizers and qubits in the code respectively.
Qubits and check nodes in the primal Tanner graphs are connected to the corresponding check nodes and qubits in the dual Tanner graphs.
All qubits in the ancilla are initialized in the  $ \vert 0\rangle$ state and then stabilizer measurements are performed for $d$ rounds, where $d$ is the code distance.

Stabilizers of the ancilla are designed such that a suitable product of them gives the logical operator we wish to measure.
Consider the product of all $X$-type stabilizers in the ancilla.
Note that, by construction, any qubit in the ancilla lattice participates in exactly an even number of $X$-type stabilizers.
Also, the qubits in support of the logical operator participate in the $X$-type stabilizers of the dual Tanner graph in the first layer.
As a result, the product of the $X$-type stabilizers of ancilla has support only on the qubits in support of the logical operator:
\begin{equation}
    \prod_{i \in \text{ancilla}} S_{X,i} = \prod_{j = 1}^4 X_{q_j} = \bar L.
\end{equation}
Therefore, the logical measurement outcome can be inferred from outcomes of $X$-type stabilizers in the ancilla lattice.
After $d$ rounds of stabilizer measurements, the qubits in the ancilla are measured in the $Z$-basis.\\

\noindent \emph{Logical Z measurement.} Logical $Z$ operator, denoted $\bar Z$, has only Pauli $X$-operators in its support.
Note that $\bar Z$ shares support with $Y$-type brick stabilizers and $Y$-type stabilizer defined on twists.
However, $\bar Z$ can be deformed so that the physical qubits in its support do not participate in a $Y$-type stabilizer.
This gives the Tanner graph structure with physical qubits in the support of $\bar Z$ and only $Z$-type checks.
Therefore, protocol presented in Ref.~\cite{Cohen2022} can be used for logical-$Z$ measurement.\\

\begin{figure}[htb]
    \centering
    \begin{subfigure}{.45\textwidth}
        \centering 
        \includegraphics[scale = .75]{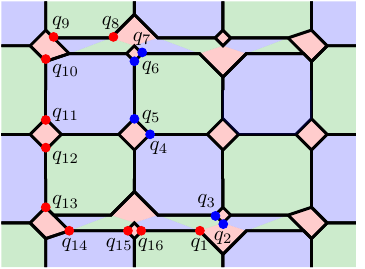}
        \subcaption{}
        \label{fig:domino-twists-modified-x-operator}   
    \end{subfigure}
    ~
    \begin{subfigure}{.45\textwidth}
        \centering 
        \includegraphics[scale = .75]{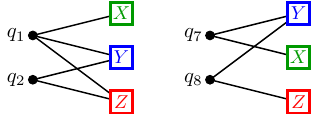} 
        \subcaption{}
        \label{fig:primal-dual-tanner-graphs-new-protocol}   
    \end{subfigure}
    ~
    \begin{subfigure}{.45\textwidth}
        \centering 
        \includegraphics[scale = .75]{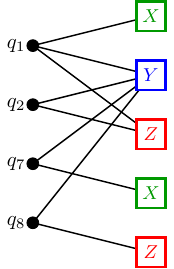} 
        \subcaption{}
        \label{fig:example-primal-tanner-graph}   
    \end{subfigure}
    ~
    \begin{subfigure}{.45\textwidth}
        \centering 
        \includegraphics[scale = .75]{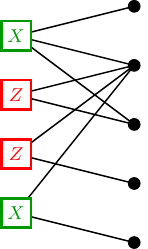} 
        \subcaption{}
        \label{fig:example-dual-tanner-graph}   
    \end{subfigure}
    \caption{(a) A representation of logical $X$ operator used for logical measurement. The vertices marked red and blue have Pauli operators $X$ and $Z$ respectively on them. (b) Tanner graphs where two qubits $q$, $q'$ participate in a $Y$-type stabilizer such that $\bar X \vert_{q,q'} = X_q Z_{q'}$. The qubits $q_1$, $q_2$, $q_7$, and $q_8$ are as shown in (a). (c) Combining the two checks corresponding to $Y$-type stabilizers into one check node. (d) The dual Tanner graph obtained from (c).}
\end{figure}

\begin{figure}[p]
    \centering
    \includegraphics[scale =.6,angle = 90]{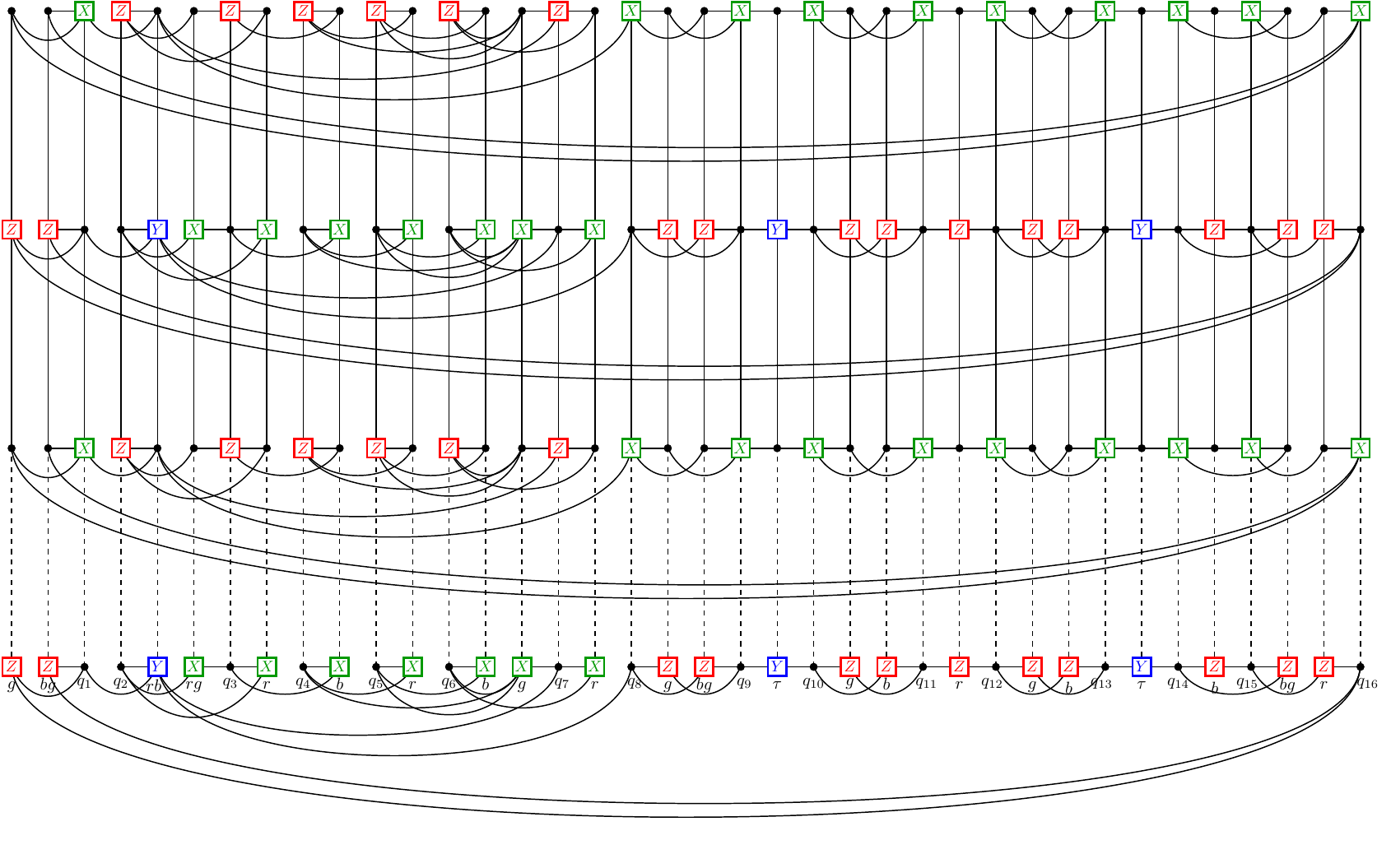}
    \caption{Protocol for measuring the logical $X$ operator in Tanner graph formalism. The sub-Tanner graph in the bottom corresponds to the code and the rest to the ancilla. Qubit numbering is same as that in Fig.~\ref{fig:domino-twists-modified-x-operator}. Labels below check nodes indicate the color of the face the particular stabilizer belongs to. For instance, the label $g$ indicates a stabilizer on green face and the label $bg$ indicates a brick stabilizer with half-bricks of color blue and green. The $Y$-stabilizer with label $rb$ corresponds to the product of two brick stabilizers $B_{rb}$. The label $\tau$ indicates twists.}
    \label{fig:domino-twists-tanner-graph}
\end{figure}

\subsection{Logical X measurement}
We now focus on measuring the logical $X$ operator, see Fig.~\ref{fig:domino-LO-X}.
The protocol in Ref.~\cite{Cohen2022} cannot be used for performing logical $X$ measurements.
The reason is that the logical $X$ operator has a mixture of Pauli $X$ and $Z$ operators on qubits in its support and also shares qubits with non-CSS stabilizers ($Y$-type twist stabilizers and red-blue brick stabilizers). 
Also, we cannot deform $\bar X$ such that is does not have support on $Y$-type stabilizers.
Therefore, we adapt the protocol in Ref.~\cite{Cohen2022} to suit logical-$X$ measurement.

We now construct the primal Tanner graph with qubits in support of $\bar X$ and stabilizers that share common support with $\bar X$.
A specific qubit node and a check node are connected if the qubit participates in the check.
Note that some qubits in support of $\bar X$ participate in $X$-, $Y$-, and $Z$-type stabilizers.
Consequently, we have three types of check nodes in the primal Tanner graph: $X$-type, $Y$-type, and $Z$-type check nodes, see Fig.~\ref{fig:primal-dual-tanner-graphs-new-protocol}.
We modify the primal Tanner graph depending on the qubits participating in $Y$-type stabilizers.
A $Y$-type stabilizer sharing two qubits $q$ and $q'$ with $\bar X$ such that $\bar X \vert_{q,q'} = X_{q} Z_{q'}$ are not individually represented. 
The $Y$-type brick stabilizers $B_{rb}^d$ have this property, see Fig.~\ref{fig:primal-dual-tanner-graphs-new-protocol}.
The reason being when we take the dual of the Tanner graph, the qubit corresponding to the $Y$-type stabilizer ($q_y$) participates in precisely one $X$ and $Z$ type stabilizer corresponding to qubits $q$ and $q'$ respectively.
When we take the product of stabilizers in the dual layers, this leads to $\bar X \prod_{q_y} Y_{q_y}$ which is undesirable.
Instead, two such $Y$-type stabilizers are jointly represented by a single $Y$-type check node, see Fig.~\ref{fig:example-primal-tanner-graph}.
This is equivalent to replacing the individual $Y$-type stabilizers by their product. 
However, if the restriction of $\bar X$ to the common support with a $Y$-type stabilizer has the same Pauli operators i.e., $\bar X \vert_{q,q'} = P_{q} P_{q'}$, $P \in \{X,Z \}$, then such $Y$-type stabilizers are represented as individual $Y$-type check nodes.
The $Y$-type stabilizer defined on twists satisfy this condition.
For example, the qubits $q_{13}$ and $q_{14}$ are common to $\bar X$ and $Y$-type twist stabilizer and $\bar X \vert_{q_{13},q_{14}} = X_{q_{13}} X_{q_{14}}$.
Therefore, $Y$-type stabilizer on twist face is retained as such in the Tanner graph.
This concludes the construction of primal Tanner graph.

The procedure for constructing the dual Tanner graph is as follows:
\begin{compactenum}[i)]
    \item Checks of $X$-, $Y$-, and $Z$-type are mapped to qubits.
    \item For a qubit $q \in \supp(\bar X)$ such that $\bar X \vert_q = P$, where $P \in \{X,Z \}$, then $q$ is mapped to a $P$-type check..
\end{compactenum}
The dual of the Tanner graph in Fig.~\ref{fig:example-primal-tanner-graph} is shown in Fig.~\ref{fig:example-dual-tanner-graph}.

The protocol for measuring the operator $\bar X$  is shown in Fig.~\ref{fig:domino-twists-tanner-graph}.
We use the qubit numbering given in Fig.~\ref{fig:domino-twists-modified-x-operator}.
A primal layer is sandwiched between two dual layers.
As before, if the weight of $\bar X$ is $w$, then there are $w - 1$ primal Tanner graphs sandwiched between $w$ dual Tanner graphs (for brevity, Fig.~\ref{fig:domino-twists-tanner-graph} shows two dual and one primal Tanner graphs).
Thus, the number of qubits used for the ancilla lattice is $O(w^2)$.
The stabilizer measurements are repeated for $d$ rounds where $d$ is the code distance.
Post stabilizer measurement, the qubits in the ancilla are measured in the $Z$ basis.\\

\begin{theorem}[X-logical measurement.]
     The protocol shown in Fig.~\ref{fig:domino-twists-tanner-graph} implements logical $X$ measurement. 
    \label{thm:logical-x-measurement}
\end{theorem}

\begin{proof}
    Consider the product of stabilizers in the dual Tanner graphs in Fig.~\ref{fig:domino-twists-tanner-graph}.
    Observe that each qubit in the ancilla participates exactly in an even number of $X$-type or $Z$-type stabilizers
    (the qubit corresponding to $Y$-stabilizer with label $rb$ participates in precisely two $X$-type and $Z$-type stabilizers).
    Also note that the qubits in support of $\bar X$ participate in the stabilizers of the dual Tanner graph.
    Therefore, the product of stabilizers of dual Tanner graphs gives $\bar X$.
\end{proof}

The protocols presented above can be used for logical qubit initialization as well.
We now move on to protocols for implementing logical gates on encoded qubits.

\section{Logical gates}	
\label{sec:clifford-gates}

In this section, we present protocols for implementing logical Clifford gates by braiding domino twists.
Prior to that, we establish one-to-one correspondence between Ising anyon model and color code with domino twists.
This correspondence is made use of during logical gate implementation.

\subsection{Mapping to Ising anyon model}
\label{sec:Ising-model}

The Ising anyon model has three components:
\begin{compactenum}[i)]
\item Particles: $1$, $\psi$ and $\sigma$, where $1$ is the vacuum, $\psi$ is fermion, and $\sigma$ is non-Abelian anyon.
\item Fusion Rules: $1 \times a = a $, $\psi \times \psi = 1$, $\sigma \times \psi = \sigma$, $\sigma \times \sigma = 1 + \psi $.
\item  $R$ and $F$ matrices.
\end{compactenum}
\vspace{2mm}

The particles in the color code with domino twists have similar behavior as the corresponding particles in the Ising model.
The particle $rx \times bz$ is a Fermion in the color code with domino twists.
Note that non-Abelian anyons have dimension greater than one.
The dimension of domino twists is $\sqrt{2}$~\cite{Kesselring2018} and hence we can think of domino twists as hosting a non-Abelian anyon.

The anyons in color code obey the corresponding fusion rules.
For instance, consider the fusion of two non-Abelian anyons.
The fusion of two non-Abelian anyons in color codes corresponds to measuring the string operator encircling a pair of twists, see Fig.~\ref{fig:domino-LO-Z}.
Since the string encircling a pair of twists is a Pauli operator, its measurement outcome is either $+1$ or $-1$.
The outcomes $\pm 1$ correspond to non-Abelian anyons fusing to either vacuum or fermion respectively in the Ising model.
The error operator that creates the particle $rx \times bz$ commutes with brick stabilizers and twists, see Fig.~\ref{fig:domino-LO-X}.
Hence, measuring brick stabilizers or twist stabilizers will not reveal the presence of this particle. 
Correspondingly, the fermion $rx \times bz$ is \textit{absorbed} by the domain wall.
The above arguments motivate us to map the elements in the color code with domino twists to Ising model.
The mapping is given in Table~\ref{tab:ising_domino}.

\begin{table}[htb]
    \centering
    \caption{Mapping between particle types in Ising anyon model and color code with domino twists.}
    \begin{tabular}{lllllllll}
    \hline
    \hline
      Ising Model  &&&&&&& & Color code with domino twist \\
      \hline
        $1$ & &&&&&&& $1$ (vacuum)\\
        \hline
        $\psi$ & &&&&&&& $rx \times bz$\\
        \hline
        $\sigma$ &&&&&&& & domino twist\\
        \hline
        \hline
    \end{tabular}
    \label{tab:ising_domino}
\end{table}

\subsection{Majorana and Fermionic operators}
A Majorana operator $c_m$ is associated with every twist $t_m$ satisfying the condition~\cite{Sarma2015}
\begin{equation}
    c_m c_n + c_n c_m = 2\delta_{mn}, c_m^2 = 1.
    \label{eqn:majorana-op-condition}
\end{equation}

\begin{figure}[htb]
    \centering\includegraphics[scale = 0.9]{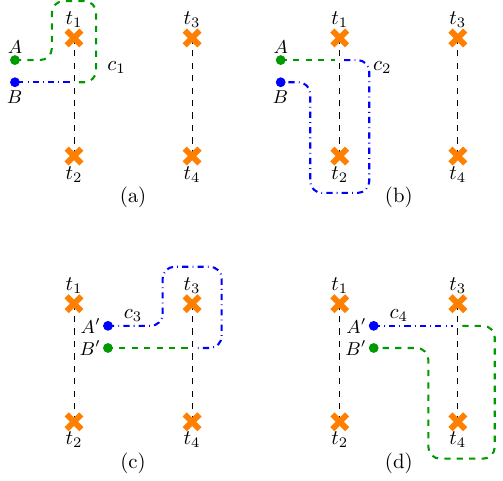}
    \caption{Majorana operators in color code with domino twists.}
    \label{fig:majorana-operators-domino}
\end{figure}

Majorana operators associated with domino twists are shown in Fig.~\ref{fig:majorana-operators-domino}.
Observe that the operators $c_1$ and $c_2$ have different Pauli operators on the vertex where they intersect and hence the two operators anticommute.
Note that operator $c_m$ does not commute with stabilizer generators as it corresponds to an open string on the lattice.
However, the product of a two Majorana operators commutes with stabilizer generators (strings corresponding to such operators encircle a pair of twists), see Appendix~\ref{sec:majorana-to-log-op-string}.
Such operators are self-adjoint and commute with all stabilizer generators.
Therefore, such operators are Fermionic~\cite{Sarma2015} and thus we can treat a twist pair as a Fermionic mode.
Closed string operators enclosing an even number of twists is the product of the Majorana operators of individual twists enclosed by the closed string.

\noindent \emph{Fermion operators.}
A set of operators $\{ \hat{a}_i \}$ are called Fermion operators if they satisfy the following properties~\cite{sakuraiA}:
\begin{equation}
    \{\hat{a}_i, \hat{a}_j \} = \{ \hat{a}_i^{\dagger}, \hat{a}_j^{\dagger} \} = 0, \text{ } \{\hat{a}_i, \hat{a}_j^{\dagger} \} = \delta_{ij}.
    \label{eqn:fermion-properties}
\end{equation}
The operators $\hat{a}$ and $\hat{a}^{\dagger}$ are called annihilation and creation operators respectively.

Denote the basis states by $\vert n \rangle$ where $n$ denotes the number of particles present in the system.
Such states are called number states.
In this particular case, $n = 0,1$ as we are associating the basis with a Fermionic mode.
We denote the basis states by $\vert0\rangle$ and $\vert1\rangle$ corresponding to the values $0$ and $1$.
The action of operators $\hat{a}$ and $\hat{a}^{\dagger}$ on basis states is given below:
\begin{equation}
	\hat{a}\vert0\rangle = 0, \quad \hat{a}\vert1\rangle = \vert0\rangle, \quad \hat{a}^{\dagger}\vert0\rangle = \vert1\rangle, \quad \hat{a}^{\dagger}\vert1\rangle = 0.	
 \label{eqn:action-of-fermion-ops}
\end{equation}

Majorana and Fermion operators can be expressed in terms of the other as follows~\cite{Kitaev2006,Roy2017}:
\begin{subequations}
    \begin{eqnarray}
    \label{eqn:majorana-to-fermion}
c_{2k - 1} = \hat{a}_k + \hat{a}_k^{\dagger}, \text{     } c_{2k} = -i(\hat{a}_k - \hat{a}_k^{\dagger}), \\
\label{eqn:fermion-to-majorana}
\hat{a}_k =  \frac{c_{2k - 1} + i c_{2k}}{2}, \text{     } \hat{a}_k^{\dagger} =  \frac{c_{2k - 1} - i c_{2k} }{2}.
\end{eqnarray}
\label{eqn:anyon-fermion-ops-mapping}
\end{subequations}
where $ k = 1,2, \dots, m$.
Thus, a creation and an annihilation operator $(\hat a_k, \hat a^{\dagger}_k)$ can be associated with a pair of twists created together.

\nix{
Define an operator $\hat{N}$ such that $\hat{N}\vert n\rangle = n\vert n\rangle$. 
The operator $\hat{N}$ is Hermitian and measures the number of particles in state $\vert n\rangle$ and is given by $\hat{N} = \hat{a}^{\dagger} \hat{a}$. 
Consider the operator $\hat{N}^2$:
\begin{equation*}
    \hat{N}^2 = \hat{a}^{\dagger}\hat{a}\hat{a}^{\dagger}\hat{a} = \hat{a}^{\dagger} (1 - \hat{a}^{\dagger}\hat{a})\hat{a} =  \hat{a}^{\dagger}\hat{a} = \hat{N}.
\end{equation*}
The operator $N$ satisfies the relation $\hat{N}^2 = \hat{N}$. 
Therefore, eigenvalues of $\hat{N}$ can take only two values: $0$ and $1$.
Consequently, the operator $\hat{a}$ acts on a Hilbert space of dimension two. 
}

\subsection{Encoding}
We associate a number basis $\{ \vert 0 \rangle, \vert 1 \rangle \}$ with a pair of twists $t_{2m-1}$ and $t_{2m}$.
The state $ \vert 0 \rangle $ ($\vert 1 \rangle$) corresponds to the absence (presence) of a fermion.
The states $ \vert0\rangle $ and $\vert1\rangle$ are $+1$ and $-1$ eigenstates of the operator $-ic_{2m-1}c_{2m}$ which is a closed string enclosing two twists $t_{2m-1}$ and $t_{2m}$. 

Suppose that the lattice has only a pair of twists.
From Eq.~\eqref{eqn:num-ind-stabs} it follows that the number of independent stabilizers is $n$.
Therefore, the dimension of codespace is one.
Let $O$ be the operator corresponding to the string encircling the twists.
The operator $O$ is not in the stabilizer group.
In this case, we can initialize the code in either $\vert 0\rangle$ or $\vert 1\rangle$ which correspond to the $\pm 1$-eigenstates of the operator $O$.
Therefore, two twists cannot encode a logical qubit~\cite{Bravyi2006,Nayak2008}. 
If we start with a lattice with zero encoded qubits, two pairs of twists are needed to encode a logical qubit.

Suppose that there are four twists $t_1$, $t_2$, $t_3$ and $t_4$ in the lattice.
The number basis used is $ \vert b_1, b_2 \rangle$ where $b_1, b_2 \in \{0,1 \}$, $b_1$ and $b_2$ are associated with twist pairs ($t_1$, $t_2$) and ($t_3$, $t_4$) respectively.
We use states $\vert 0,0 \rangle$ and $\vert 1,1 \rangle$ which correspond to the subspace where fusion outcome is vacuum.
The states $\vert 0,0 \rangle$ and $\vert 1,1 \rangle$ are $ \pm 1 $-eigenstates of the operator $-ic_1c_2$.
This encoding can be extended to more than two twist pairs.
If there are six twists, $t_{1}, \dots, t_{6}$, the number basis is $\vert b_1, b_2, b_3 \rangle$ where $b_1, b_2, b_3 \in \{0,1 \} $.
As before, we choose the subspace where the fusion outcome is zero.
This subspace is spanned by $\vert 0,0,0 \rangle$, $\vert 0,1,1 \rangle$, $\vert 1,0,1 \rangle$, $\vert 1,1,0 \rangle$. 
Note that these states are the $+1$-eigenstates of the operator $ (-ic_1c_2)(-ic_3c_4)(-ic_5c_6)$. 
We call this basis as logical basis. 

We use the relation in Eq.~\eqref{eqn:majorana-to-fermion} to analyse the action of operator $-ic_2 c_3$ and $-i c_1 c_3$ on the logical basis.
Note that $-ic_2c_3 =   (\hat{a}_1^{\dagger} - \hat{a}_1) (\hat{a}_2 + \hat{a}_2^{\dagger})$ and $-ic_1c_3 = -i(\hat{a}_1 + \hat{a}_1^{\dagger})(\hat{a}_2^{\dagger} + \hat{a}_2) $. 
We have,
\begin{subequations}
    \begin{eqnarray}
        	-ic_2c_3\vert0,0\rangle = \vert1,1\rangle, & & -ic_2c_3\vert1,1\rangle = - \vert0,0\rangle\\
         -ic_1c_3\vert0,0\rangle = -i\vert1,1\rangle, & &  -ic_1c_3\vert1,1\rangle = -i\vert0,0\rangle
    \end{eqnarray}
\end{subequations}
From the action of operators $-ic_2c_3$ and $-ic_1c_3$, it is clear that they act as $Y$ and $X$ operators up to an overall phase. 
Correcting for the phase, we get $c_1c_3$ and $c_2c_3$ to be logical $X$ and $Y$ operators respectively. 
Therefore, the operators $c_1c_3$, $c_2c_3$,  and $-ic_1c_2$ are logical $X$, $Y$ and $Z$ operators respectively.
We take the strings encircling twists ($t_1$, $t_2$) and ($t_1$, $t_3$) to be the logical $Z$ and logical $X$ operators respectively, see Fig.~\ref{fig:lo-domino}.\\

\noindent \emph{Mapping number states to codewords.}
The Hilbert space of $n$ Fermionic modes was identified with the Hilbert space of $n$ qubits in Ref.~\cite{Bravyi2002}.
Here, we show that we can map the number states to the code space.
Recall that the number state indicates the presence or absence of fermions in a given pair of twists.
Such states can be mapped to codestates.
For instance, consider the state $\vert 0,0 \rangle $.
This state corresponds to a system where the outcome of fusing non-Abelian anyon pairs is zero.
In code space, this state corresponds to the common $+1$-eigenstate of stabilizers and the logical $Z$ operator.
Therefore, $\vert 0,0 \rangle \longrightarrow \vert \bar 0 \rangle_{code}$.
One can similarly argue that $\vert 1,1 \rangle \longrightarrow \vert \bar 1 \rangle_{code}$ where $\vert \bar 1 \rangle$ is the $+1$-eigenstate of stabilizers and and a $-1$-eigenstate of the logical $Z$ operator.
When there are three pairs of twists in the lattice, the state $\vert 0,0,0 \rangle$ corresponds to the state $\vert \bar 0 \rangle_1 \otimes \vert \bar 0 \rangle_2$.
The state $\vert \bar 0 \rangle_1 \otimes \vert \bar 0 \rangle_2 $ is the $+1$-eigenstate of all stabilizer generators, $\bar Z_1$, and $\bar Z_2$.
Other states and their mapping to the states in codespace are given in Table~\ref{tab:fusion-code-mapping}.

\begin{table}[htb]
    \centering
        \caption{Mapping between number states and codestates.}
    \renewcommand{\arraystretch}{1.5}
    \begin{tabular}{ccccccc}
    \hline
    \hline
        & Number states & Codestates & Stabilizers & $\bar Z_1$ & $\bar Z_2$ & \\
        \hline
        &  $\vert 0,0,0 \rangle$ & $\vert \bar 0\rangle_1 \otimes \vert \bar 0 \rangle_2$ & $+1$ & $+1$ & $+1$  & \\
        &  $\vert 0,1,1 \rangle$ & $\vert \bar 0\rangle_1 \otimes \vert \bar 1 \rangle_2$ & $+1$ & $+1$ & $-1$ &  \\
        &  $\vert 1,1,0 \rangle$ & $\vert \bar 1\rangle_1 \otimes \vert \bar 0 \rangle_2$ & $+1$ & $-1$ & $+1$ &  \\
        &  $\vert 1,0,1 \rangle$ & $\vert \bar 1\rangle_1 \otimes \vert \bar 1 \rangle_2$ & $+1$ & $-1$ & $-1$ &  \\
        \hline 
        \hline
    \end{tabular}
    \label{tab:fusion-code-mapping}
\end{table}

\subsection{Gates}

The world-line of twists without braiding is shown Fig.~\ref{fig:encoding}. 
We use $\lfloor t / 3 \rfloor$-encoding and therefore three twists are used to encode a logical qubit.
The closed string encompassing twists $t_{1}$, $t_{2}$ and $t_{2}$, $t_{3}$ corresponds to the logical $Z$ and $X$ operators respectively, see Fig.~\ref{fig:lo-domino}.
Braiding a pair of twists encircled by a logical operator, say $\bar{Z}$, leaves the logical operator unchanged and exchanges the other two logical operators, $\bar{X} \longleftrightarrow \bar{Y}$~\cite{GowdaSarvepalli2020, GowdaSarvepalli2021}.
This is reminiscent of Pauli  $Z$ rotation by $\pi /  2$.
We can conclude that braiding twists has the effect of $\pi / 2$ rotation on the code space up to an overall phase~\cite{Nayak2008}. 
The above argument holds when twists encircled by strings corresponding to logical $X$ or $Y$ operators are braided.

\begin{figure}[htb]
  \centering
      \begin{subfigure}{0.225\textwidth}
         \centering
         \includegraphics{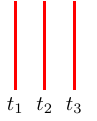}
         \subcaption{}
        \label{fig:encoding}
      \end{subfigure}
      ~
      \begin{subfigure}{0.225\textwidth}
         \centering
         \includegraphics{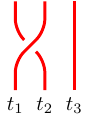}
         \subcaption{}
        \label{fig:phase}
      \end{subfigure}
      ~
      \begin{subfigure}{0.225\textwidth}
         \centering
         \includegraphics{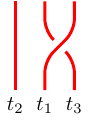}
         \subcaption{ }
         \label{fig:hadamard}
      \end{subfigure}
      ~
       \begin{subfigure}{.225\textwidth}
         \centering
         \includegraphics{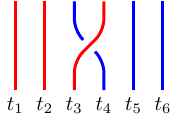}
         \subcaption{}
        \label{fig:entangling-gate}
      \end{subfigure}
\caption{ (a) Recall that we use $\lfloor t / 3 \rfloor$-encoding. So, three twists are used to encode a logical qubit.  Logical-$Z$ operator is represented by the string enclosing twists $t_1$ and $t_2$ and logical-$X$ operator by the string enclosing $t_1$ and $t_3$. (b) Braiding twists $t_1$ and $t_2$ counterclockwise implements phase gate. (c) Braiding twists $t_1$ and $t_3$ implements $X$-rotation by $\pi / 2$ gate on the codespace. Note that the positions of twists $t_1$ and $t_2$ are interchanged for clarity. (d) Protocol to implement an entangling gate up to single qubit rotations. Twists $t_{1}$-$t_{3}$ and $t_{4}$-$t_{6}$ are used to encode logical qubits on which the entangling gate acts. }
\label{fig:braid}
\end{figure}

We analyze the effect of braiding twist pairs using number states and Fermion operators.
Since number states have one to one correspondence with code states, the results directly carry over to the code space.
Recall that logical $X$ and $Z$ operators are $\bar{Z}_1 = -i c_1c_2, \quad \bar{X}_1 =  c_1 c_3$.
The unitary operator resulting from braiding twists $k$ and $l$ can be described as~\cite{Bravyi2006}
	\begin{equation}
		\Theta(k,l) = \exp \left( -i \frac{\pi}{4} (-ic_kc_l) \right) = \frac{1}{\sqrt{2}} \left( I - c_kc_l \right).
        \label{eqn:braiding-k-l-twists-effect}
	\end{equation}

\paragraph*{$Z$ rotation by $\pi / 2 $.} 
$Z$ rotation by $\pi / 2$ is implemented by braiding twists $t_1$ and $t_2$ which are encircled by the string corresponding to logical $Z$ operator.
From Eq.~\ref{eqn:braiding-k-l-twists-effect}, it follows that 
\begin{equation}
     \Theta(1,2) = \frac{1}{\sqrt{2}} \left( I - c_1c_2 \right).
\end{equation}
Note that $c_1 c_2 = -i (\hat{a}_1 + \hat{a}_1^\dagger)(\hat{a}_1 - \hat{a}_1^\dagger)$ and hence
\begin{equation}
    \Theta(1,2) =  \frac{1}{\sqrt{2}} [I + i (\hat{a}_1^\dagger \hat{a}_1 - \hat{a}_1 \hat{a}_1^\dagger )  ].
\end{equation}
Using this, we derive the transformation of the basis states $ \vert0,0\rangle$ and $\vert 1,1 \rangle$ under the action of the unitary $\Theta(1,2)$:
\begin{subequations}
\begin{eqnarray}
\Theta(1,2) \vert0,0\rangle &=& \frac{1}{\sqrt{2}} \left[I + i (\hat{a}_1^\dagger \hat{a}_1 - \hat{a}_1 \hat{a}_1^\dagger )  \right] \vert0,0\rangle  = \frac{1-i}{\sqrt{2}} \vert0,0\rangle,  \\
\Theta(1,2) \vert1,1\rangle &=& \frac{1}{\sqrt{2}} \left[I + i (\hat{a}_1^\dagger \hat{a}_1 - \hat{a}_1 \hat{a}_1^\dagger )  \right] \vert1,1\rangle  = \frac{1+i}{\sqrt{2}} \vert1,1\rangle.
\end{eqnarray}
\label{eqn:braiding-twists-1-2}
\end{subequations}
where we have used Eq.~\eqref{eqn:action-of-fermion-ops}.
In the basis of $\vert0,0\rangle $ and $\vert1,1\rangle$, the unitary $\Theta(1,2)$ has the form
\begin{equation}
    \Theta(1,2) = \left( \begin{matrix}
e^{-i\pi / 4} & 0\\ 0 & e^{i\pi / 4}
\end{matrix} \right)\\ 
 = e^{-i \pi/4} \left( \begin{matrix}
1 & 0\\ 0 & i
\end{matrix} \right) = e^{-i\pi / 4}S = R_Z(\pi /2).
\label{eqn:phase-rz-rotation}
\end{equation}
It can be seen that the unitary $\Theta(1,2)$ is the phase gate up to an overall phase.
The braiding is shown in  Fig.~\ref{fig:phase}.

\paragraph*{$X$ rotation by $\pi / 2 $.} 
The $X$ rotation by $\pi / 2$ is implemented by braiding twists $t_1$ and $t_3$.
Note that
\begin{equation}
   \Theta(1,3) = \exp\left( - i\frac{\pi}{4} \bar{X}\right) = \exp\left( - i\frac{\pi}{4} c_1 c_3 \right) = \frac{\left( I - ic_1 c_3 \right)}{\sqrt{2}} 
   \label{eqn:theta-1-3}
\end{equation}
We have, $-ic_1c_3 = -i(\hat{a}_1 + \hat{a}_1^{\dagger})(\hat{a}_2^{\dagger} + \hat{a}_2) = -i( \hat{a}_1 \hat{a}_2^\dagger + \hat{a}_1 \hat{a}_2 + \hat{a}_1^\dagger \hat{a}_2^\dagger + \hat{a}_1^\dagger \hat{a}_2)$.
The effect of the unitary $\Theta(1,3)$ on the basis states $\vert0,0\rangle$ and $\vert1,1\rangle$ is as below:
\begin{subequations}
    \begin{eqnarray}
        \Theta(1,3) \vert0,0\rangle &=& \frac{1}{\sqrt{2}} \left[ I -i(\hat{a}_1 + \hat{a}_1^{\dagger})(\hat{a}_2^{\dagger} + \hat{a}_2)\right] \vert 0,0 \rangle = \frac{1}{\sqrt{2}} \left(\vert 0,0\rangle  -i\vert 1,1 \rangle \right),\\
        \Theta(1,3) \vert 1,1 \rangle &=& \frac{1}{\sqrt{2}} \left[ I -i(\hat{a}_1 + \hat{a}_1^{\dagger})(\hat{a}_2^{\dagger} + \hat{a}_2)\right] \vert 1,1 \rangle = \frac{1}{\sqrt{2}} \left(\vert1,1\rangle  -i\vert0,0\rangle \right).
    \end{eqnarray}
    \label{eqn:effct-of-1-3}
\end{subequations}
In the basis of $\vert0,0\rangle$ and $\vert1,1\rangle$, the operator $\Theta(1,3)$ can be written as
\begin{equation}
    \Theta(1,3) = \frac{1}{\sqrt{2}}\left( \begin{matrix}
 1 &  -i \\  -i &  1
\end{matrix} \right) = R_X(\pi / 2).
\end{equation}
The braiding rule is shown in Fig.~\ref{fig:hadamard}. Hadamard gate is implemented by sandwiching $R_X(\pi / 2)$ gate between $R_Z(\pi / 2)$ gates.
Specifically, $H = e^{-i \alpha} R_Z(\pi /2) R_X(\pi /2) R_Z(\pi /2) $~\cite{Nielsen2010}.

\paragraph*{Entangling gate.}
We consider the effect of braiding twists $t_3$ and $t_4$ on the logical basis states which is given by the matrix $\Theta(3,4) = \frac{1}{\sqrt{2}} (I - c_3c_4)$.
It can be shown that the string encircling twists  $t_3$ and $t_4$ is $\bar{Z}_1 \bar{Z}_2$~\cite{GowdaSarvepalli2021}.
Recall that the  states $\vert0,0,0\rangle$, $\vert1,0,1\rangle$ and $\vert0,1,1\rangle$, $\vert1,1,0\rangle$  are $+1$ and $-1$ eigensates of the operator $-ic_3c_4$ respectively.
The action of $\Theta(3,4)$ on the basis states is obtained by expressing Majorana operators $c_3$ and $c_4$ in terms of fermion operators $\hat{a}_2$ and $\hat{a}_2^\dagger$ and is as follows:
\begin{subequations}
\begin{eqnarray}
\Theta(3,4)\vert0,0,0\rangle &=& \frac{1}{\sqrt{2}} (1-i) \vert0,0,0\rangle \\
\Theta(3,4)\vert0,1,1\rangle &=& \frac{1}{\sqrt{2}} (1+ i) \vert0,1,1\rangle \\
\Theta(3,4)\vert1,1,0\rangle &=& \frac{1}{\sqrt{2}} (1+ i) \vert1,1,0\rangle\\
\Theta(3,4)\vert1,0,1\rangle &=& \frac{1}{\sqrt{2}} (1-i) \vert1,0,1\rangle 
\end{eqnarray}
\end{subequations}
In logical basis spanned by $\vert 0,0,0 \rangle$, $\vert 0,1,1\rangle$, $\vert 1,1,0 \rangle$, and $\vert 1,0,1 \rangle$ (in this order) the unitary $\Theta(3,4)$ is a diagonal matrix with elements $e^{-i \pi /4}$, $e^{i \pi /4}$, $e^{i \pi /4}$, and $e^{-i \pi /4}$ respectively.
Up to an overall phase, we can write $\Theta(3,4)$ as
\begin{equation}
    \Theta(3,4) = \exp{ \left(-i \frac{\pi}{4} \bar{Z}_1 \bar{Z}_2 \right)} = \diag \{ 1,i,i,1 \}.
    \label{eqn:entangling-gate}
\end{equation}

The gate specified by Eq.~\ref{eqn:entangling-gate} is controlled-phase gate up to conjugate-phase gates on logical qubits.
This braiding implements the following transformation on logical operators:
$\bar{Z}_1 \rightarrow \bar{Z}_1$,
$\bar{X}_1 \rightarrow \bar{Y}_1 \bar{Z}_2$, 
$\bar{Z}_2 \rightarrow \bar{Z}_2 $,
$\bar{X}_2 \rightarrow \bar{Z}_1 \bar{Y}_2$.
The results of this section is summarized in the theorem below:\\

\begin{theorem}[Logical Clifford gates.]
All logical Clifford gates are realized in color codes with domino twists by braiding.\\
\label{thm:encoded-Cliffords}
\end{theorem}

This concludes our discussion on implementing logical Clifford gates by braiding. 
To be able to implement a general quantum circuit, we need a non-Clifford gate as well.
We briefly discuss the implementation of a non-Clifford gate that makes use of specially prepared states called magic states.\\

\noindent \emph{Non-Clifford gate.} A non-Clifford gate cannot be implemented by braiding.
Special states, called magic states, are required for implementing non-Clifford gates.
An often used gate is the $T$ gate defined as
\begin{equation}
    T = \diag\{1, e^{i \pi / 4} \}.
\end{equation}
The magic state used to implement this is $\vert m\rangle = 2^{-1/2} ( \vert 0\rangle + e^{i \pi / 4} \vert 1\rangle )$.
Using magic state encoded in an ancilla, the $T$-gate is implemented by joint $Z$-measurement of logical qubit and ancilla qubit followed by a CNOT gate with logical qubit as the control.
The ancilla is measured after CNOT and phase gate is applied on the logical qubit if the ancilla measurement outcome is $-1$~\cite{Bravyi2006}.
This concludes our discussion on implementing logical gates.

\section{Conclusion}
In this paper, we have presented a systematic way to introduce domino twists in color codes.
We have presented the code parameters for the constructed codes.
Also, we have presented protocols for carrying out logical measurements using ancilla qubits.
We have shown that all logical Clifford gates can be performed by braiding.
A direction for future research could be constructing efficient decoders for color codes with twists. 
Another fruitful direction could be to extend the protocols presented for logical measurements to non-CSS quantum codes.

\backmatter


\bmhead{Acknowledgements}
The author would like to thank Prof. Pradeep Sarvepalli and Guillaume Dauphinias for helpful discussions. 

\section*{Declarations}


\begin{itemize}
\item Conflict of interest: The author has no relevant conflicts of interest to disclose.
\end{itemize}



\bigskip





\begin{appendices}
\section{Stabilizer generators as strings}
In this section, we present strings corresponding to a subset of stabilizer generators.
Note that stabilizers do not produce syndromes.
As a result, the strings corresponding to stabilizers are closed strings.
\begin{figure}[htb]
    \centering
    \begin{subfigure}{.45\textwidth}
        \centering
        \includegraphics[scale = .75]{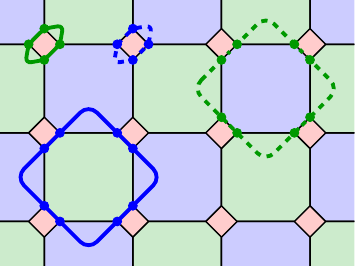}
        \caption{}
        \label{fig:normal-stabilizers-strings}
    \end{subfigure}
    ~
    \begin{subfigure}{.45\textwidth}
        \centering
        \includegraphics[scale = .65]{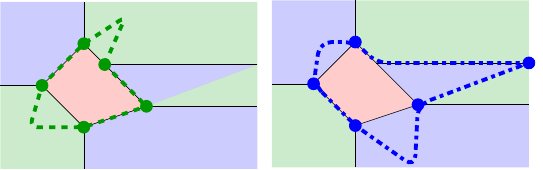}
        \caption{}
        \label{fig:twists-X-Y-stabilizers-strings}
    \end{subfigure}
    \caption{String description for stabilizers: (a) normal faces (b) twists.}
    \label{fig:stabilizer-strings}
\end{figure}
String description for normal stabilizer generators and twists is shown in Fig.~\ref{fig:normal-stabilizers-strings} and Fig.~\ref{fig:twists-X-Y-stabilizers-strings} respectively.

\section{Proof of Lemma~\ref{lm:encoded-qubits-domino}}
\label{sec:proof-encoded-qubits}

Let $\tau = t / 2$ be the number of twist pairs and let $\ell$ be the length of the virtual path (i.e., the path formed by common edges to the half-bricks).
For each twist pair, there are $\ell - 1$ two valent vertices and the total number of two valent vertices is $\tau (\ell - 1 )$.
Since the graph is embedded on two dimensional plane, we have $n + f - e = 2$, where $n$ and $f$ are number of vertices and faces respectively (note that $f$ also includes the domain wall as one face and the unbounded green face).
Counting the degree of each vertex, we get $2e = 3[n - \tau (\ell - 1 )] + 2\tau (\ell - 1 ) = 3n + \tau - \tau \ell$, where $e$ is the number of edges in the lattice.
Using the above result, we get, $f = (n/2) - (\tau / 2) (\ell - 1) + 2 $.
Removing the face counted for the domain wall for every twist pair (the number of such faces is $\tau$), we get $f =  (n/ 2) - (\tau / 2) (\ell + 1) + 2$.
The number of stabilizers defined on normal faces and twists is $2f$.
Including the brick stabilizers (which are $\tau \ell$ in number) and removing dependencies, we the number of independent stabilizers as
\begin{equation}
    s = 2f + \tau \ell - 3 = n - \tau + 1 = n - \frac{t}{2} + 1.
    \label{eqn:num-ind-stabs}
\end{equation}
Therefore, the number of encoded qubits is $k = \frac{t}{2} - 1$, where $t$ is the number of twists in the lattice.

\begin{figure}[htb]
    \centering
    \begin{subfigure}{.45\textwidth}
        \centering
        \includegraphics[scale = .75]{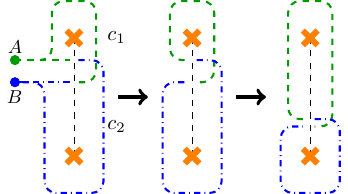}
        \label{fig:combining-majorana-operators-1}
        \subcaption{}
    \end{subfigure}
    ~
    \begin{subfigure}{.45\textwidth}
        \centering
        \includegraphics[scale = .75]{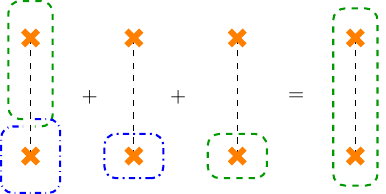}
        \label{fig:combining-majorana-operators-2}
        \subcaption{}
    \end{subfigure}
    \caption{Combining two Majorana operators to get a logical operator.}
    \label{fig:combining-majorana-operators}
\end{figure}

\section{Combining two Majorana operators to get a logical operator}
\label{sec:majorana-to-log-op-string}
The operators $c_m$ are open string operators and therefore do not belong to $ \mathcal{C}( \mathcal{S})$, the centralizer of stabilizer group. 
However, the operator $-ic_mc_{m+1}$ is closed  loop operator and is in $\mathcal{C}(\mathcal{S})$, see Fig.~\ref{fig:combining-majorana-operators}. 
The operator $-ic_mc_{m+1}$ is called Fermionic string operator (as its measurement outcome indicates the presence or absence of Fermions). 
Measurement outcome of this operator gives the total charge of the twists (fusion outcome). 
If the outcome is $+1$, then the twists fuse to vacuum and to fermion $\psi$ if the outcome is $-1$.
\end{appendices}




\begin{thebibliography}{34}
\ifx \bisbn   \undefined \def \bisbn  #1{ISBN #1}\fi
\ifx \binits  \undefined \def \binits#1{#1}\fi
\ifx \bauthor  \undefined \def \bauthor#1{#1}\fi
\ifx \batitle  \undefined \def \batitle#1{#1}\fi
\ifx \bjtitle  \undefined \def \bjtitle#1{#1}\fi
\ifx \bvolume  \undefined \def \bvolume#1{\textbf{#1}}\fi
\ifx \byear  \undefined \def \byear#1{#1}\fi
\ifx \bissue  \undefined \def \bissue#1{#1}\fi
\ifx \bfpage  \undefined \def \bfpage#1{#1}\fi
\ifx \blpage  \undefined \def \blpage #1{#1}\fi
\ifx \burl  \undefined \def \burl#1{\textsf{#1}}\fi
\ifx \doiurl  \undefined \def \doiurl#1{\url{https://doi.org/#1}}\fi
\ifx \betal  \undefined \def \betal{\textit{et al.}}\fi
\ifx \binstitute  \undefined \def \binstitute#1{#1}\fi
\ifx \binstitutionaled  \undefined \def \binstitutionaled#1{#1}\fi
\ifx \bctitle  \undefined \def \bctitle#1{#1}\fi
\ifx \beditor  \undefined \def \beditor#1{#1}\fi
\ifx \bpublisher  \undefined \def \bpublisher#1{#1}\fi
\ifx \bbtitle  \undefined \def \bbtitle#1{#1}\fi
\ifx \bedition  \undefined \def \bedition#1{#1}\fi
\ifx \bseriesno  \undefined \def \bseriesno#1{#1}\fi
\ifx \blocation  \undefined \def \blocation#1{#1}\fi
\ifx \bsertitle  \undefined \def \bsertitle#1{#1}\fi
\ifx \bsnm \undefined \def \bsnm#1{#1}\fi
\ifx \bsuffix \undefined \def \bsuffix#1{#1}\fi
\ifx \bparticle \undefined \def \bparticle#1{#1}\fi
\ifx \barticle \undefined \def \barticle#1{#1}\fi
\bibcommenthead
\ifx \bconfdate \undefined \def \bconfdate #1{#1}\fi
\ifx \botherref \undefined \def \botherref #1{#1}\fi
\ifx \url \undefined \def \url#1{\textsf{#1}}\fi
\ifx \bchapter \undefined \def \bchapter#1{#1}\fi
\ifx \bbook \undefined \def \bbook#1{#1}\fi
\ifx \bcomment \undefined \def \bcomment#1{#1}\fi
\ifx \oauthor \undefined \def \oauthor#1{#1}\fi
\ifx \citeauthoryear \undefined \def \citeauthoryear#1{#1}\fi
\ifx \endbibitem  \undefined \def \endbibitem {}\fi
\ifx \bconflocation  \undefined \def \bconflocation#1{#1}\fi
\ifx \arxivurl  \undefined \def \arxivurl#1{\textsf{#1}}\fi
\csname PreBibitemsHook\endcsname

\bibitem[\protect\citeauthoryear{Bombin and Martin-Delgado}{2006}]{Bombin2006}
\begin{barticle}
\bauthor{\bsnm{Bombin}, \binits{H.}},
\bauthor{\bsnm{Martin-Delgado}, \binits{M.A.}}:
\batitle{Topological quantum distillation}.
\bjtitle{Phys. Rev. Lett.}
\bvolume{97},
\bfpage{180501}
(\byear{2006})
\doiurl{10.1103/PhysRevLett.97.180501}
\end{barticle}
\endbibitem

\bibitem[\protect\citeauthoryear{Fowler}{2011}]{Fowler2011}
\begin{barticle}
\bauthor{\bsnm{Fowler}, \binits{A.G.}}:
\batitle{Two-dimensional color-code quantum computation}.
\bjtitle{Phys. Rev. A}
\bvolume{83},
\bfpage{042310}
(\byear{2011})
\doiurl{10.1103/PhysRevA.83.042310}
\end{barticle}
\endbibitem

\bibitem[\protect\citeauthoryear{Landahl et~al.}{2011}]{Landahl2011}
\begin{botherref}
\oauthor{\bsnm{Landahl}, \binits{A.J.}},
\oauthor{\bsnm{Anderson}, \binits{J.T.}},
\oauthor{\bsnm{Rice}, \binits{P.R.}}:
Fault-tolerant quantum computing with color codes.
arXiv:1108.5738v1 [quant-ph]
(2011)
\end{botherref}
\endbibitem

\bibitem[\protect\citeauthoryear{Kesselring et~al.}{2018}]{Kesselring2018}
\begin{barticle}
\bauthor{\bsnm{Kesselring}, \binits{M.S.}},
\bauthor{\bsnm{Pastawski}, \binits{F.}},
\bauthor{\bsnm{Eisert}, \binits{J.}},
\bauthor{\bsnm{Brown}, \binits{B.J.}}:
\batitle{The boundaries and twist defects of the color code and their applications to topological quantum computation}.
\bjtitle{{Quantum}}
\bvolume{2},
\bfpage{101}
(\byear{2018})
\doiurl{10.22331/q-2018-10-19-101}
\end{barticle}
\endbibitem

\bibitem[\protect\citeauthoryear{Gowda and Sarvepalli}{2021}]{GowdaSarvepalli2021}
\begin{barticle}
\bauthor{\bsnm{Gowda}, \binits{M.G.}},
\bauthor{\bsnm{Sarvepalli}, \binits{P.K.}}:
\batitle{Color codes with twists: Construction and universal-gate-set implementation}.
\bjtitle{Phys. Rev. A}
\bvolume{104},
\bfpage{012603}
(\byear{2021})
\doiurl{10.1103/PhysRevA.104.012603}
\end{barticle}
\endbibitem

\bibitem[\protect\citeauthoryear{{Landahl} and {Ryan-Anderson}}{2014}]{Landahl2014}
\begin{botherref}
\oauthor{\bsnm{{Landahl}}, \binits{A.J.}},
\oauthor{\bsnm{{Ryan-Anderson}}, \binits{C.}}:
Quantum computing by color-code lattice surgery.
arXiv:1407.5103v1 [quant-ph]
(2014)
\end{botherref}
\endbibitem

\bibitem[\protect\citeauthoryear{Mukai et~al.}{2020}]{Mukai2020}
\begin{barticle}
\bauthor{\bsnm{Mukai}, \binits{H.}},
\bauthor{\bsnm{Sakata}, \binits{K.}},
\bauthor{\bsnm{Devitt}, \binits{S.J.}},
\bauthor{\bsnm{Wang}, \binits{R.}},
\bauthor{\bsnm{Zhou}, \binits{Y.}},
\bauthor{\bsnm{Nakajima}, \binits{Y.}},
\bauthor{\bsnm{Tsai}, \binits{J.-S.}}:
\batitle{Pseudo-2d superconducting quantum computing circuit for the surface code: proposal and preliminary tests}.
\bjtitle{New Journal of Physics}
\bvolume{22}(\bissue{4}),
\bfpage{043013}
(\byear{2020})
\doiurl{10.1088/1367-2630/ab7d7d}
\end{barticle}
\endbibitem

\bibitem[\protect\citeauthoryear{Andersen et~al.}{2020}]{Andersen2020}
\begin{barticle}
\bauthor{\bsnm{Andersen}, \binits{C.K.}},
\bauthor{\bsnm{Remm}, \binits{A.}},
\bauthor{\bsnm{Lazar}, \binits{S.}},
\bauthor{\bsnm{Krinner}, \binits{S.}},
\bauthor{\bsnm{Lacroix}, \binits{N.}},
\bauthor{\bsnm{Norris}, \binits{G.J.}},
\bauthor{\bsnm{Gabureac}, \binits{M.}},
\bauthor{\bsnm{Eichler}, \binits{C.}},
\bauthor{\bsnm{Wallraff}, \binits{A.}}:
\batitle{Repeated quantum error detection in a surface code}.
\bjtitle{Nature Physics}
\bvolume{16}(\bissue{8}),
\bfpage{875}--\blpage{880}
(\byear{2020})
\doiurl{10.1038/s41567-020-0920-y}
\end{barticle}
\endbibitem

\bibitem[\protect\citeauthoryear{Ryan-Anderson et~al.}{2021}]{Ryan-Anderson2021}
\begin{barticle}
\bauthor{\bsnm{Ryan-Anderson}, \binits{C.}},
\bauthor{\bsnm{Bohnet}, \binits{J.G.}},
\bauthor{\bsnm{Lee}, \binits{K.}},
\bauthor{\bsnm{Gresh}, \binits{D.}},
\bauthor{\bsnm{Hankin}, \binits{A.}},
\bauthor{\bsnm{Gaebler}, \binits{J.P.}},
\bauthor{\bsnm{Francois}, \binits{D.}},
\bauthor{\bsnm{Chernoguzov}, \binits{A.}},
\bauthor{\bsnm{Lucchetti}, \binits{D.}},
\bauthor{\bsnm{Brown}, \binits{N.C.}},
\bauthor{\bsnm{Gatterman}, \binits{T.M.}},
\bauthor{\bsnm{Halit}, \binits{S.K.}},
\bauthor{\bsnm{Gilmore}, \binits{K.}},
\bauthor{\bsnm{Gerber}, \binits{J.A.}},
\bauthor{\bsnm{Neyenhuis}, \binits{B.}},
\bauthor{\bsnm{Hayes}, \binits{D.}},
\bauthor{\bsnm{Stutz}, \binits{R.P.}}:
\batitle{Realization of real-time fault-tolerant quantum error correction}.
\bjtitle{Phys. Rev. X}
\bvolume{11},
\bfpage{041058}
(\byear{2021})
\doiurl{10.1103/PhysRevX.11.041058}
\end{barticle}
\endbibitem

\bibitem[\protect\citeauthoryear{Chen et~al.}{2021}]{Chen2021}
\begin{barticle}
\bauthor{\bsnm{Chen}, \binits{Z.}},
\bauthor{\bsnm{Satzinger}, \binits{K.J.}},
\bauthor{\bsnm{Atalaya}, \binits{J.}},
\bauthor{\bsnm{Korotkov}, \binits{A.N.}},
\bauthor{\bsnm{Dunsworth}, \binits{A.}},
\bauthor{\bsnm{Sank}, \binits{D.}},
\bauthor{\bsnm{Quintana}, \binits{C.}},
\bauthor{\bsnm{McEwen}, \binits{M.}},
\bauthor{\bsnm{Barends}, \binits{R.}},
\bauthor{\bsnm{Klimov}, \binits{P.V.}},
\bauthor{\bsnm{Hong}, \binits{S.}},
\bauthor{\bsnm{Jones}, \binits{C.}},
\bauthor{\bsnm{Petukhov}, \binits{A.}},
\bauthor{\bsnm{Kafri}, \binits{D.}},
\bauthor{\bsnm{Demura}, \binits{S.}},
\bauthor{\bsnm{Burkett}, \binits{B.}},
\bauthor{\bsnm{Gidney}, \binits{C.}},
\bauthor{\bsnm{Fowler}, \binits{A.G.}},
\bauthor{\bsnm{Paler}, \binits{A.}},
\bauthor{\bsnm{Putterman}, \binits{H.}},
\bauthor{\bsnm{Aleiner}, \binits{I.}},
\bauthor{\bsnm{Arute}, \binits{F.}},
\bauthor{\bsnm{Arya}, \binits{K.}},
\bauthor{\bsnm{Babbush}, \binits{R.}},
\bauthor{\bsnm{Bardin}, \binits{J.C.}},
\bauthor{\bsnm{Bengtsson}, \binits{A.}},
\bauthor{\bsnm{Bourassa}, \binits{A.}},
\bauthor{\bsnm{Broughton}, \binits{M.}},
\bauthor{\bsnm{Buckley}, \binits{B.B.}},
\bauthor{\bsnm{Buell}, \binits{D.A.}},
\bauthor{\bsnm{Bushnell}, \binits{N.}},
\bauthor{\bsnm{Chiaro}, \binits{B.}},
\bauthor{\bsnm{Collins}, \binits{R.}},
\bauthor{\bsnm{Courtney}, \binits{W.}},
\bauthor{\bsnm{Derk}, \binits{A.R.}},
\bauthor{\bsnm{Eppens}, \binits{D.}},
\bauthor{\bsnm{Erickson}, \binits{C.}},
\bauthor{\bsnm{Farhi}, \binits{E.}},
\bauthor{\bsnm{Foxen}, \binits{B.}},
\bauthor{\bsnm{Giustina}, \binits{M.}},
\bauthor{\bsnm{Greene}, \binits{A.}},
\bauthor{\bsnm{Gross}, \binits{J.A.}},
\bauthor{\bsnm{Harrigan}, \binits{M.P.}},
\bauthor{\bsnm{Harrington}, \binits{S.D.}},
\bauthor{\bsnm{Hilton}, \binits{J.}},
\bauthor{\bsnm{Ho}, \binits{A.}},
\bauthor{\bsnm{Huang}, \binits{T.}},
\bauthor{\bsnm{Huggins}, \binits{W.J.}},
\bauthor{\bsnm{Ioffe}, \binits{L.B.}},
\bauthor{\bsnm{Isakov}, \binits{S.V.}},
\bauthor{\bsnm{Jeffrey}, \binits{E.}},
\bauthor{\bsnm{Jiang}, \binits{Z.}},
\bauthor{\bsnm{Kechedzhi}, \binits{K.}},
\bauthor{\bsnm{Kim}, \binits{S.}},
\bauthor{\bsnm{Kitaev}, \binits{A.}},
\bauthor{\bsnm{Kostritsa}, \binits{F.}},
\bauthor{\bsnm{Landhuis}, \binits{D.}},
\bauthor{\bsnm{Laptev}, \binits{P.}},
\bauthor{\bsnm{Lucero}, \binits{E.}},
\bauthor{\bsnm{Martin}, \binits{O.}},
\bauthor{\bsnm{McClean}, \binits{J.R.}},
\bauthor{\bsnm{McCourt}, \binits{T.}},
\bauthor{\bsnm{Mi}, \binits{X.}},
\bauthor{\bsnm{Miao}, \binits{K.C.}},
\bauthor{\bsnm{Mohseni}, \binits{M.}},
\bauthor{\bsnm{Montazeri}, \binits{S.}},
\bauthor{\bsnm{Mruczkiewicz}, \binits{W.}},
\bauthor{\bsnm{Mutus}, \binits{J.}},
\bauthor{\bsnm{Naaman}, \binits{O.}},
\bauthor{\bsnm{Neeley}, \binits{M.}},
\bauthor{\bsnm{Neill}, \binits{C.}},
\bauthor{\bsnm{Newman}, \binits{M.}},
\bauthor{\bsnm{Niu}, \binits{M.Y.}},
\bauthor{\bsnm{O'Brien}, \binits{T.E.}},
\bauthor{\bsnm{Opremcak}, \binits{A.}},
\bauthor{\bsnm{Ostby}, \binits{E.}},
\bauthor{\bsnm{Pat{\'o}}, \binits{B.}},
\bauthor{\bsnm{Redd}, \binits{N.}},
\bauthor{\bsnm{Roushan}, \binits{P.}},
\bauthor{\bsnm{Rubin}, \binits{N.C.}},
\bauthor{\bsnm{Shvarts}, \binits{V.}},
\bauthor{\bsnm{Strain}, \binits{D.}},
\bauthor{\bsnm{Szalay}, \binits{M.}},
\bauthor{\bsnm{Trevithick}, \binits{M.D.}},
\bauthor{\bsnm{Villalonga}, \binits{B.}},
\bauthor{\bsnm{White}, \binits{T.}},
\bauthor{\bsnm{Yao}, \binits{Z.J.}},
\bauthor{\bsnm{Yeh}, \binits{P.}},
\bauthor{\bsnm{Yoo}, \binits{J.}},
\bauthor{\bsnm{Zalcman}, \binits{A.}},
\bauthor{\bsnm{Neven}, \binits{H.}},
\bauthor{\bsnm{Boixo}, \binits{S.}},
\bauthor{\bsnm{Smelyanskiy}, \binits{V.}},
\bauthor{\bsnm{Chen}, \binits{Y.}},
\bauthor{\bsnm{Megrant}, \binits{A.}},
\bauthor{\bsnm{Kelly}, \binits{J.}},
\bauthor{\bsnm{AI}, \binits{G.Q.}}:
\batitle{Exponential suppression of bit or phase errors with cyclic error correction}.
\bjtitle{Nature}
\bvolume{595}(\bissue{7867}),
\bfpage{383}--\blpage{387}
(\byear{2021})
\doiurl{10.1038/s41586-021-03588-y}
\end{barticle}
\endbibitem

\bibitem[\protect\citeauthoryear{Erhard et~al.}{2021}]{Erhard2021}
\begin{barticle}
\bauthor{\bsnm{Erhard}, \binits{A.}},
\bauthor{\bsnm{Poulsen~Nautrup}, \binits{H.}},
\bauthor{\bsnm{Meth}, \binits{M.}},
\bauthor{\bsnm{Postler}, \binits{L.}},
\bauthor{\bsnm{Stricker}, \binits{R.}},
\bauthor{\bsnm{Stadler}, \binits{M.}},
\bauthor{\bsnm{Negnevitsky}, \binits{V.}},
\bauthor{\bsnm{Ringbauer}, \binits{M.}},
\bauthor{\bsnm{Schindler}, \binits{P.}},
\bauthor{\bsnm{Briegel}, \binits{H.J.}},
\bauthor{\bsnm{Blatt}, \binits{R.}},
\bauthor{\bsnm{Friis}, \binits{N.}},
\bauthor{\bsnm{Monz}, \binits{T.}}:
\batitle{Entangling logical qubits with lattice surgery}.
\bjtitle{Nature}
\bvolume{589}(\bissue{7841}),
\bfpage{220}--\blpage{224}
(\byear{2021})
\doiurl{10.1038/s41586-020-03079-6}
\end{barticle}
\endbibitem

\bibitem[\protect\citeauthoryear{Marques et~al.}{2022}]{Marques2022}
\begin{barticle}
\bauthor{\bsnm{Marques}, \binits{J.F.}},
\bauthor{\bsnm{Varbanov}, \binits{B.M.}},
\bauthor{\bsnm{Moreira}, \binits{M.S.}},
\bauthor{\bsnm{Ali}, \binits{H.}},
\bauthor{\bsnm{Muthusubramanian}, \binits{N.}},
\bauthor{\bsnm{Zachariadis}, \binits{C.}},
\bauthor{\bsnm{Battistel}, \binits{F.}},
\bauthor{\bsnm{Beekman}, \binits{M.}},
\bauthor{\bsnm{Haider}, \binits{N.}},
\bauthor{\bsnm{Vlothuizen}, \binits{W.}},
\bauthor{\bsnm{Bruno}, \binits{A.}},
\bauthor{\bsnm{Terhal}, \binits{B.M.}},
\bauthor{\bsnm{DiCarlo}, \binits{L.}}:
\batitle{Logical-qubit operations in an error-detecting surface code}.
\bjtitle{Nature Physics}
\bvolume{18}(\bissue{1}),
\bfpage{80}--\blpage{86}
(\byear{2022})
\doiurl{10.1038/s41567-021-01423-9}
\end{barticle}
\endbibitem

\bibitem[\protect\citeauthoryear{Hilder et~al.}{2022}]{Hilder2022}
\begin{barticle}
\bauthor{\bsnm{Hilder}, \binits{J.}},
\bauthor{\bsnm{Pijn}, \binits{D.}},
\bauthor{\bsnm{Onishchenko}, \binits{O.}},
\bauthor{\bsnm{Stahl}, \binits{A.}},
\bauthor{\bsnm{Orth}, \binits{M.}},
\bauthor{\bsnm{Lekitsch}, \binits{B.}},
\bauthor{\bsnm{Rodriguez-Blanco}, \binits{A.}},
\bauthor{\bsnm{M\"uller}, \binits{M.}},
\bauthor{\bsnm{Schmidt-Kaler}, \binits{F.}},
\bauthor{\bsnm{Poschinger}, \binits{U.G.}}:
\batitle{Fault-tolerant parity readout on a shuttling-based trapped-ion quantum computer}.
\bjtitle{Phys. Rev. X}
\bvolume{12},
\bfpage{011032}
(\byear{2022})
\doiurl{10.1103/PhysRevX.12.011032}
\end{barticle}
\endbibitem

\bibitem[\protect\citeauthoryear{Kitaev}{2003}]{Kitaev2003}
\begin{barticle}
\bauthor{\bsnm{Kitaev}, \binits{A.Y.}}:
\batitle{Fault-tolerant quantum computation by anyons}.
\bjtitle{Annals of Physics}
\bvolume{303}(\bissue{1}),
\bfpage{2}--\blpage{30}
(\byear{2003})
\doiurl{10.1016/S0003-4916(02)00018-0}
\end{barticle}
\endbibitem

\bibitem[\protect\citeauthoryear{Bombin}{2010}]{Bombin2010}
\begin{barticle}
\bauthor{\bsnm{Bombin}, \binits{H.}}:
\batitle{Topological order with a twist: Ising anyons from an abelian model}.
\bjtitle{Phys. Rev. Lett.}
\bvolume{105},
\bfpage{030403}
(\byear{2010})
\doiurl{10.1103/PhysRevLett.105.030403}
\end{barticle}
\endbibitem

\bibitem[\protect\citeauthoryear{Hastings and Geller}{2015}]{Hastings2015}
\begin{barticle}
\bauthor{\bsnm{Hastings}, \binits{M.B.}},
\bauthor{\bsnm{Geller}, \binits{A.}}:
\batitle{Reduced space-time and time costs using dislocation codes and arbitrary ancillas}.
\bjtitle{Quantum Info. Comput.}
\bvolume{15}(\bissue{11-12}),
\bfpage{962}--\blpage{986}
(\byear{2015})
\end{barticle}
\endbibitem

\bibitem[\protect\citeauthoryear{Brown et~al.}{2017}]{Brown2017}
\begin{barticle}
\bauthor{\bsnm{Brown}, \binits{B.J.}},
\bauthor{\bsnm{Laubscher}, \binits{K.}},
\bauthor{\bsnm{Kesselring}, \binits{M.S.}},
\bauthor{\bsnm{Wootton}, \binits{J.R.}}:
\batitle{Poking holes and cutting corners to achieve clifford gates with the surface code}.
\bjtitle{Phys. Rev. X}
\bvolume{7},
\bfpage{021029}
(\byear{2017})
\doiurl{10.1103/PhysRevX.7.021029}
\end{barticle}
\endbibitem

\bibitem[\protect\citeauthoryear{Lavasani and Barkeshli}{2018}]{Lavasani2018}
\begin{barticle}
\bauthor{\bsnm{Lavasani}, \binits{A.}},
\bauthor{\bsnm{Barkeshli}, \binits{M.}}:
\batitle{Low overhead clifford gates from joint measurements in surface, color, and hyperbolic codes}.
\bjtitle{Phys. Rev. A}
\bvolume{98},
\bfpage{052319}
(\byear{2018})
\doiurl{10.1103/PhysRevA.98.052319}
\end{barticle}
\endbibitem

\bibitem[\protect\citeauthoryear{Gowda and Sarvepalli}{2020}]{GowdaSarvepalli2020}
\begin{barticle}
\bauthor{\bsnm{Gowda}, \binits{M.G.}},
\bauthor{\bsnm{Sarvepalli}, \binits{P.K.}}:
\batitle{Quantum computation with generalized dislocation codes}.
\bjtitle{Phys. Rev. A}
\bvolume{102},
\bfpage{042616}
(\byear{2020})
\doiurl{10.1103/PhysRevA.102.042616}
\end{barticle}
\endbibitem

\bibitem[\protect\citeauthoryear{Gowda and Sarvepalli}{2022}]{GowdaSarvepalli2022}
\begin{barticle}
\bauthor{\bsnm{Gowda}, \binits{M.G.}},
\bauthor{\bsnm{Sarvepalli}, \binits{P.K.}}:
\batitle{Quantum computation with charge-and-color-permuting twists in qudit color codes}.
\bjtitle{Phys. Rev. A}
\bvolume{105},
\bfpage{022621}
(\byear{2022})
\doiurl{10.1103/PhysRevA.105.022621}
\end{barticle}
\endbibitem

\bibitem[\protect\citeauthoryear{Bombin}{2011}]{Bombin2011}
\begin{barticle}
\bauthor{\bsnm{Bombin}, \binits{H.}}:
\batitle{Clifford gates by code deformation}.
\bjtitle{New Journal of Physics}
\bvolume{13}(\bissue{4}),
\bfpage{043005}
(\byear{2011})
\end{barticle}
\endbibitem

\bibitem[\protect\citeauthoryear{You and Wen}{2012}]{Yu-Wen2012}
\begin{barticle}
\bauthor{\bsnm{You}, \binits{Y.-Z.}},
\bauthor{\bsnm{Wen}, \binits{X.-G.}}:
\batitle{Projective non-abelian statistics of dislocation defects in a $\mathbb{Z}_{N}$ rotor model}.
\bjtitle{Phys. Rev. B}
\bvolume{86},
\bfpage{161107}
(\byear{2012})
\doiurl{10.1103/PhysRevB.86.161107}
\end{barticle}
\endbibitem

\bibitem[\protect\citeauthoryear{Cohen et~al.}{2022}]{Cohen2022}
\begin{barticle}
\bauthor{\bsnm{Cohen}, \binits{L.Z.}},
\bauthor{\bsnm{Kim}, \binits{I.H.}},
\bauthor{\bsnm{Bartlett}, \binits{S.D.}},
\bauthor{\bsnm{Brown}, \binits{B.J.}}:
\batitle{Low-overhead fault-tolerant quantum computing using long-range connectivity}.
\bjtitle{Science Advances}
\bvolume{8}(\bissue{20}),
\bfpage{1717}
(\byear{2022})
\doiurl{10.1126/sciadv.abn1717}
{\href{https://arxiv.org/abs/https://www.science.org/doi/pdf/10.1126/sciadv.abn1717}{{https://www.science.org/doi/pdf/10.1126/sciadv.abn1717}}}
\end{barticle}
\endbibitem

\bibitem[\protect\citeauthoryear{Bravyi and Kitaev}{2002}]{Bravyi2002}
\begin{barticle}
\bauthor{\bsnm{Bravyi}, \binits{S.B.}},
\bauthor{\bsnm{Kitaev}, \binits{A.Y.}}:
\batitle{Fermionic quantum computation}.
\bjtitle{Annals of Physics}
\bvolume{298}(\bissue{1}),
\bfpage{210}--\blpage{226}
(\byear{2002})
\doiurl{10.1006/aphy.2002.6254}
\end{barticle}
\endbibitem

\bibitem[\protect\citeauthoryear{Bravyi}{2006}]{Bravyi2006}
\begin{barticle}
\bauthor{\bsnm{Bravyi}, \binits{S.}}:
\batitle{Universal quantum computation with the $\ensuremath{\nu}= 5 \slash 2$ fractional quantum hall state}.
\bjtitle{Phys. Rev. A}
\bvolume{73},
\bfpage{042313}
(\byear{2006})
\doiurl{10.1103/PhysRevA.73.042313}
\end{barticle}
\endbibitem

\bibitem[\protect\citeauthoryear{Scheppe and Pak}{2022}]{Scheppe2022}
\begin{barticle}
\bauthor{\bsnm{Scheppe}, \binits{A.D.}},
\bauthor{\bsnm{Pak}, \binits{M.V.}}:
\batitle{Complete description of fault-tolerant quantum gate operations for topological majorana qubit systems}.
\bjtitle{Phys. Rev. A}
\bvolume{105},
\bfpage{012415}
(\byear{2022})
\doiurl{10.1103/PhysRevA.105.012415}
\end{barticle}
\endbibitem

\bibitem[\protect\citeauthoryear{Zheng et~al.}{2015}]{Zheng2015}
\begin{barticle}
\bauthor{\bsnm{Zheng}, \binits{H.}},
\bauthor{\bsnm{Dua}, \binits{A.}},
\bauthor{\bsnm{Jiang}, \binits{L.}}:
\batitle{Demonstrating non-abelian statistics of majorana fermions using twist defects}.
\bjtitle{Phys. Rev. B}
\bvolume{92},
\bfpage{245139}
(\byear{2015})
\doiurl{10.1103/PhysRevB.92.245139}
\end{barticle}
\endbibitem

\bibitem[\protect\citeauthoryear{Sarma et~al.}{2015}]{Sarma2015}
\begin{barticle}
\bauthor{\bsnm{Sarma}, \binits{S.D.}},
\bauthor{\bsnm{Freedman}, \binits{M.}},
\bauthor{\bsnm{Nayak}, \binits{C.}}:
\batitle{Majorana zero modes and topological quantum computation}.
\bjtitle{npj Quantum Information}
\bvolume{1}(\bissue{1}),
\bfpage{15001}
(\byear{2015})
\doiurl{10.1038/npjqi.2015.1}
\end{barticle}
\endbibitem

\bibitem[\protect\citeauthoryear{Nielsen and Chuang}{2010}]{Nielsen2010}
\begin{bbook}
\bauthor{\bsnm{Nielsen}, \binits{M.A.}},
\bauthor{\bsnm{Chuang}, \binits{I.L.}}:
\bbtitle{Quantum Computation and Quantum Information: 10th Anniversary Edition}.
\bpublisher{Cambridge University Press},
\blocation{Cambridge}
(\byear{2010}).
\doiurl{10.1017/CBO9780511976667}
\end{bbook}
\endbibitem

\bibitem[\protect\citeauthoryear{{Lidar} and {Brun (Eds.)}}{2013}]{LidarBrun2013}
\begin{bbook}
\bauthor{\bsnm{{Lidar}}, \binits{D.}},
\bauthor{\bsnm{{Brun (Eds.)}}, \binits{T.A.}}:
\bbtitle{Quantum Error Correction}.
\bpublisher{Cambridge University Press},
\blocation{New York}
(\byear{2013}).
\doiurl{10.1017/CBO9781139034807}
\end{bbook}
\endbibitem

\bibitem[\protect\citeauthoryear{Sakurai}{1987}]{sakuraiA}
\begin{bbook}
\bauthor{\bsnm{Sakurai}, \binits{J.J.}}:
\bbtitle{Advanced Quantum Mechanics}.
\bsertitle{Addison-Wesley Series in Advanced Physics}.
\bpublisher{Addison-Wesley},
\blocation{London}
(\byear{1987})
\end{bbook}
\endbibitem

\bibitem[\protect\citeauthoryear{Kitaev}{2006}]{Kitaev2006}
\begin{barticle}
\bauthor{\bsnm{Kitaev}, \binits{A.}}:
\batitle{Anyons in an exactly solved model and beyond}.
\bjtitle{Annals of Physics}
\bvolume{321}(\bissue{1}),
\bfpage{2}--\blpage{111}
(\byear{2006})
\doiurl{10.1016/j.aop.2005.10.005} .
\bcomment{January Special Issue}
\end{barticle}
\endbibitem

\bibitem[\protect\citeauthoryear{{Roy} and {DiVincenzo}}{2017}]{Roy2017}
\begin{botherref}
\oauthor{\bsnm{{Roy}}, \binits{A.}},
\oauthor{\bsnm{{DiVincenzo}}, \binits{D.P.}}:
{Topological Quantum Computing}.
arXiv:1701.05052 [quant-ph]
(2017)
\end{botherref}
\endbibitem

\bibitem[\protect\citeauthoryear{Nayak et~al.}{2008}]{Nayak2008}
\begin{barticle}
\bauthor{\bsnm{Nayak}, \binits{C.}},
\bauthor{\bsnm{Simon}, \binits{S.H.}},
\bauthor{\bsnm{Stern}, \binits{A.}},
\bauthor{\bsnm{Freedman}, \binits{M.}},
\bauthor{\bsnm{Das~Sarma}, \binits{S.}}:
\batitle{Non-abelian anyons and topological quantum computation}.
\bjtitle{Rev. Mod. Phys.}
\bvolume{80},
\bfpage{1083}--\blpage{1159}
(\byear{2008})
\doiurl{10.1103/RevModPhys.80.1083}
\end{barticle}
\endbibitem

\end{thebibliography}



\end{document}